\documentclass[twocolumn]{IEEEtran}
\usepackage[final]{graphicx}

\usepackage{amsmath,epsfig,amssymb,verbatim,amsopn,cite,multirow}
\usepackage{balance}
\usepackage{multirow}
\usepackage{footnote}
\usepackage{algorithm}
\usepackage{algpseudocode}
\usepackage[usenames,dvipsnames]{color}
\usepackage[all]{xy}  
\usepackage{url}
\usepackage{subcaption}
\usepackage{graphicx}
\usepackage{array}
\usepackage{amsthm}

\usepackage[nodisplayskipstretch]{setspace}
\newtheorem{Remark}{Remark}

  {\proof}{\proofend}
\newtheorem{proposition}{Proposition}

\newcommand{\qn}{{\bf n}}

\newcommand{\qs}{{\bf s}}

\newcommand{\qx}{{\bf x}}

\newcommand{\qA}{{\bf A}}

\newcommand{\qD}{{\bf D}}

\newcommand{\qG}{{\bf G}}
\newcommand{\qH}{{\bf H}}
\newcommand{\qI}{{\bf I}}

\newcommand{\qP}{{\bf P}}
\newcommand{\qQ}{{\bf Q}}

\newcommand{\qU}{{\bf U}}
\newcommand{\qV}{{\bf V}}
\newcommand{\qW}{{\bf W}}

\newcommand{\qY}{{\bf Y}}

\DeclareMathOperator*{\argmax}{arg\,max}

\newcommand{\cm}{\mathcal{C}_m}
\newcommand{\cn}{\mathcal{C}_n}

\newcommand{\Ntx}{\bar{N}}
\newcommand{\Nrx}{N}
\newcommand{\Kmax}{K_{\mathrm{max}}}
\newcommand{\Mk}{\mathcal{M}_k}
\newcommand{\Mkp}{\mathcal{M}_{k'}}
\newcommand{\Km}{\mathcal{K}_m}
\newcommand{\Kmp}{\mathcal{K}_{m'}}

\newcommand{\Csm}{\mathcal{C}_m}
\newcommand{\Usm}{\mathcal{U}_m}

\newcommand{\Morder}{M_{\mathrm{order}}}
\newcommand{\etkpmn}{\tilde{\eta}_{k',\cm, \cn}}
\newcommand{\etkmn}{\tilde{\eta}_{k,\cm, \cn}}

\newcommand{\tieta}{\boldsymbol {\tilde{\eta} }}
\newcommand{\Kset}{\boldsymbol{\mathcal {K}}}

\newcommand{\Pb}{P}
\newcommand{\Sn}{\sigma^2}
\newcommand{\Snb}{\sigma^2}
\newcommand{\Pu}{P_\mathrm{u}}

\newcommand{\tauu}{\tau_{\mathrm{u}}}

\title{\fontsize{0.83cm}{1cm}\selectfont   Fronthaul-Aware User-Centric Generalized Cell-Free Massive MIMO Systems}
\author{Zahra Mobini,~\IEEEmembership{Member,~IEEE,} Ahmet Hasim Gokceoglu,  Li Wang,~\IEEEmembership{Member,~IEEE,} Gunnar Peters\\ and Hien Quoc Ngo,~\IEEEmembership{Fellow,~IEEE}}

\allowdisplaybreaks
\allowdisplaybreaks
\begin{document}

\bstctlcite{IEEEexample:BSTcontrol}
\maketitle

\begin{abstract} We consider fronthaul-limited  generalized
zero-forcing-based cell-free massive multiple-input multiple-output  (CF-mMIMO) systems with multiple-antenna users and multiple-antenna access points (APs) relying on both cooperative beamforming (CB) and user-centric  (UC) clustering.
The proposed framework is very general and
can be degenerated into different special cases, such as pure CB/pure UC clustering, or fully centralized CB/fully distributed beamforming.  We comprehensively analyze the spectral efficiency (SE) performance of the system  wherein the users use the minimum mean-squared error-based successive interference cancellation (MMSE-SIC) scheme to detect the desired signals. Specifically, we formulate an  optimization problem for the user
association and power control for maximizing the sum SE. The formulated problem is under per-AP transmit power  and  fronthaul constraints, and is based on only long-term channel state information (CSI). The challenging formulated problem is transformed into tractable form and a novel algorithm  is proposed to solve it using
minorization maximization (MM)  technique.  We  analyze the trade-offs provided by the CF-mMIMO system with different number of CB clusters, hence highlighting the importance of the appropriate choice of CB design for different system setups. Numerical results show that for  the centralized CB, the proposed power optimization provides nearly  $59\%$ improvement in the average sum SE over the heuristic approach, and $312\%$ improvement, when the distributed beamforming is employed. 

\let\thefootnote\relax\footnotetext{
Z. Mobini and H. Q. Ngo   are with the Centre for Wireless Innovation (CWI), Queen's University Belfast, BT3 9DT Belfast, U.K. (email:\{zahra.mobini, hien.ngo\}@qub.ac.uk). H. Q. Ngo is also with the Department of Electronic Engineering, Kyung Hee University, Yongin-si,
Gyeonggi-do 17104, Republic of Korea.
A. gokceoglu1,  L. Wang, and G. Peters are with the  Huawei’s Sweden Research Center, Stockholm, Sweden (e-mail:
\{ahmet.hasim.gokceoglu1, leo.li.wang, gunnar.peters\}@huawei.com). (\textit{Corresponding authors: Hien Quoc Ngo}.)
}
\end{abstract}
\section{Introduction}

Cell-free massive multiple-input multiple-output (CF-mMIMO) is a promising wireless networking platform that can effectively cater to the pervasive connectivity needs and escalating demands for data traffic in the realm of beyond fifth generation (5G). The key idea of CF-mMIMO is to deploy  many geographically distributed access points (APs)  in the
coverage area without dividing the area into disjoint cells. The APs coordinate through multiple central processing units (CPUs) and coherently serve multiple users in the same time-frequency resources~\cite{hien:2017:wcom,Ngo:PROC:2024}. However, in practical implementations, the gains of CF-mMIMO systems might saturate with the number of users and coordinating  APs due to the significant  overhead  required to acquire data, knowledge of the channel state information (CSI), and beamforming matrix over the fronthaul network and also due to the  channel estimation errors. Moreover, serving users by far APs may not be reasonable due to the fact that they consume valuable radio resources but provide users with low improvement in the signal power.  

As a remedy, the concept of dynamic cooperation clusters has been integrated into CF-mMIMO systems, gaining significant attention in recent years, particularly with the emergence of user-centric (UC) CF mMIMO~\cite{Mohammadali:survey:2024,Buzzi:WCOM:2020}. In UC CF-mMIMO  a small set of distributed (single or few antenna) APs jointly serve a subset of users with assumption of high density of antennas with respect to the number of users. UC CF-mMIMO exhibits a remarkable capability for maximizing the efficient utilization of limited power and bandwidth resources, delivering performance that is close to that of the canonical cell-free network in terms of achievable data rate. Moreover, it demands reduced fronthaul/backhaul bandwidth and complexity, rendering it  scalable for practical implementations~\cite{Buzzi:WCOM:2020}.
The benefits of dynamic resource allocation and user scheduling for the UC CF-mMIMO networks were investigated in~\cite{Den:WCOM:2021} where the total downlink achievable rate is maximized under the user's rate requirement.  The two-stage proposed framework in~\cite{Den:WCOM:2021} first partitions users into groups and then groups are scheduled on different frequencies. The corresponding optimization problems were solved by using semi-definite relaxation (SDR) and sequential convex approximation (SCA)-based approach, respectively. In~\cite{Ammar:TWC:2021} user scheduling, power allocation, and beamforming were optimized for  UC CF-mMIMO  networks given the limited number of users to maximize sum rate  for both the coherent and non-coherent transmission modes. The authors in~\cite{Ammar:TWC:2021} resorted to tools from fractional programming, block coordinate descent, and compressive sensing  to construct the beamforming weights and user scheduling algorithms. Total energy efficiency of UC CF-mMIMO with AP selection and power control was investigated in~\cite{Hien:TGCN:2018,Dong:TVT:2019} where it was shown that power allocation algorithm together with the AP selection can notably enhance  the energy efficiency. In~\cite{Hien:TGCN:2018} power control with received-power-based AP selection and channel-quality-based AP selection were studied to decrease the power consumption caused by the fronthaul links. The proposed algorithms in~\cite{Hien:TGCN:2018}  rely on heuristics to specify AP selections.  The authors in~\cite{Dong:TVT:2019} jointly optimized the downlink power control and the active APs (considering ON/OFF activity for the APs) to minimize the total transceiver power consumption under the per-user minimum downlink ergodic spectral efficiency (SE) constraint.  Joint user association and power allocation schemes were proposed in \cite{Jiang:2018:JSAC} for cell-free visible light communication (VLC) networks to address load balancing and power control issues. Meanwhile, different optimization schemes tailored to the varying sizes of CF-mMIMO systems were introduced in \cite{Hao:IOT:2024}.  Learning-based approaches were also employed as an alternative to model-based optimization methods for designing optimal power control~\cite{Sun:2018:TSP}, user association schemes~\cite{Xu:JSASP:2023}, and precoding design~\cite{Lee:GLOBECOM:2020}.

The ongoing research efforts on CF-mMIMO   systems have mainly focused  on non-coordinating beamforming strategies such as  matched filtering (MF)~\cite{hien:2017:wcom},  local zero forcing (ZF)~\cite{Interdonato:TWC:2020, Hao:IOT:2024},  and the combination of centralized ZF  with  maximum-ratio transmission (MRT)~\cite{Hien:JCOM:2021}  to reduce the fronthaul overhead and complexity. However, the performance of CF-mMIMO systems can be considerably improved by using cooperative beamforming (CB) among the APs.
There are two ways to implement  CB 1) centralized  beamforming, and 2) distributed beamforming. In the centralized beamforming schemes  all the APs send local channel estimates  to the CPU through fronthaul links. The CPU computes  beamforming  matrix, and then sends the optimized beamforming to the APs. This requires  huge CSI exchange between the APs and the CPU and more importantly high massive computational complexity for the large number of APs/users regime.  Nevertheless, modern CPUs equipped with many  processors   and hence distributed algorithms for beamforming  could be designed to leverage these multi-core processors.
CF-mMIMO with centralized coordinated ZF-based beamforming designs  has been studied in~\cite{nayebi:2017:wcom, LIU:2020:WCOM}, while centralized minimum mean-squared error-based (MMSE) design has been investigated in~\cite{Emil:WCOM:2020} with different levels of coordination. In contrast, deep learning-based coordinated beamforming design has been proposed in~\cite{Lee:GLOBECOM:2020} for traditional cellular systems, focusing on multicell downlink scenarios with rate-limited CSI exchange between base stations (BSs).

It is imperative to note that research on CF-mMIMO with CB and/or UC clustering is still in its infancy, with numerous key challenges yet to be addressed~\cite{Mohammadi:2024:TCOM}, such as following:
\emph{i)} From a performance perspective, fully centralized CB is optimal but impractical for many real-world scenarios due to its scalability issues, requiring extensive cooperation and a costly fronthaul/backhaul network. Conversely, fully distributed (non-coordinating) beamforming, where the beamforming design is handled locally at the APs, is scalable, and has low implementation demands but yields mediocre performance.
Therefore,  studying the  trade-offs between the performance, scalability, and degree of cooperation among the APs in general CF-mMIMO systems with UC clustering for different application scenarios  is meaningful and important.
\emph{ii)} Furthermore,   most studies tend to investigate simple system setup where the users are equipped with single antenna. However, in practice,  users  can be equipped with
multiple antennas to improve system reliability due to the higher diversity gain~\cite{Mai:2020:TCOM}.  The extension of the CF-mMIMO to the case with multiple-antenna users is not trivial, especially for the CF-mMIMO with no downlink channel estimation in which channel-dependent combining schemes cannot be used at the user sides. 
\emph{iii)} In addition, most of the above mentioned studies   consider the case in which the APs are managed by one or several CPUs to which they are connected through an infinite fronthaul/backhaul capacity links.  However, the assumption of infinite fronthaul is not realistic in practice. In fact, the limited fronthaul capacity  constitute one of the most key challenges in practical CF-mMIMO systems.
The performance of  a  CF-mMIMO system with single-antenna users and non-coordinating beamforming design  in the uplink  and downlink,  taking account the fronthaul constraint, has been investigated in~\cite{Bashar:JCOM:2021,Guenach:JCOM:2021}, and~\cite{Xia:TWC:2021}, respectively.
Coordinated beamforming design  
for a reconfigurable intelligent surface (RIS)-aided CF-mMIMO system with single-antenna users under CSI
uncertainties at the transmitter and the capacity-limited fronthaul links was investigated in~\cite{Yao:TCOM:2023}.
Moreover, distributed   resource allocation algorithms for user scheduling and beamforming in a UC CF-mMIMO system were studied in~\cite{Ammar:TWC:2022}. However,  the authors in~\cite{Ammar:TWC:2022} assumed single-antenna users and infinite fronthaul capacity links. Additionally, resource allocation in \cite{Ammar:TWC:2022} relied on instantaneous CSI rather than statistical CSI. Therefore, all designs must be quickly recomputed whenever the small-scale fading coefficients change. Moreover, they necessitate instantaneous CSI knowledge at the users, leading to significant  overhead in systems with many users.  We would like to highlight that~\cite{Mai:2020:TCOM} pursued a performance analysis of CF-mMIMO systems with multi-antenna users, while the  CB and UC  designs were ignored, and the impact of limited fronthaul capacity links was not investigated either. To the best of the authors’ knowledge, the performance of fronthaul-limited CF-mMIMO, which relies on both UC and CB alongside multiple-antenna users, has not been studied. Motivated by filling the above-mentioned knowledge gaps in the literature, we consider a   fronthaul-limited UC generalized CF-mMIMO systems with multiple antennas at both APs and users, where  the available APs are divided into multiple disjoint CB clusters. APs in each CB cluster perform ZF-based CB and send data symbols to a subset (not all) of the users in the system. Then, the SE of the system is investigated comprehensively. 

The key contributions of this paper are summarized as follows:
\begin{itemize}
\item We provide an analytical framework for a fronthaul limited    CF-mMIMO system with multiple-antenna users/APs and the notion of joint CB and UC clustering wherein the users use the minimum mean-squared error-based successive interference cancellation (MMSE-SIC) scheme to detect the desired signals.  The proposed CF-mMIMO framework is very general and can cover different CF-mMIMO scenarios. In particular, it can be degenerated to different special cases such as conventional CF-mMIMO with pure CB,  pure UC clustering,   fully centralized beamforming, and   fully distributed beamforming. 
\item We formulate a novel optimization problem for the user association and power control of the  generalized CF-mMIMO system  for maximizing the  sum SE. The formulated problem is under per-AP transmit power    and fronthaul constraints.
We solve the sum-SE maximization problem by casting the original problem into two sub-problems, namely P1) power optimization given user-association  and P2) user association given power optimization.

\item We show that the power optimization in P1 is non-convex and thus we propose a novel   algorithm  to solve the challenging formulated   problem. In particular, we resort to sequential optimization framework, also known as minorization maximization (MM) approach. By using this iterative framework, we  reformulate the original non-convex power optimization problem into a series of auxiliary convex subproblems that can be easily solved by algorithms that have  low
computational complexity and their solutions converge to a locally optimal solution of the original problem. Following P1, we propose a scheme based on channel gain criterion to address the user association problem in P2.

 \item Numerical results confirm that our joint optimization
approach significantly outperforms the heuristic approaches. The simulation results also show the different trade-offs between the performance  and degree of cooperation among the APs in general CF-mMIMO systems. 
 More precisely, the proposed power allocation scheme yields  increasing benefits as the size of CB clusters (number of coordinating APs) decreases, i.e., when interference cancellation via beamforming is limited. Furthermore, as the beamforming capability of APs improves (with more antennas per AP), the proposed power allocation significantly enhances fully distributed beamforming performance, reducing the need for beamforming coordination among APs. In addition, our results indicate that under practical fronthaul limitations, the proposed power allocation combined with simple distributed CB architecture can even outperform a computationally-heavy fully centralized CB without power control.
\end{itemize}

\begin{figure}[t]
	\centering
	\vspace{0em}
	\includegraphics[width=65mm, height=60mm]{  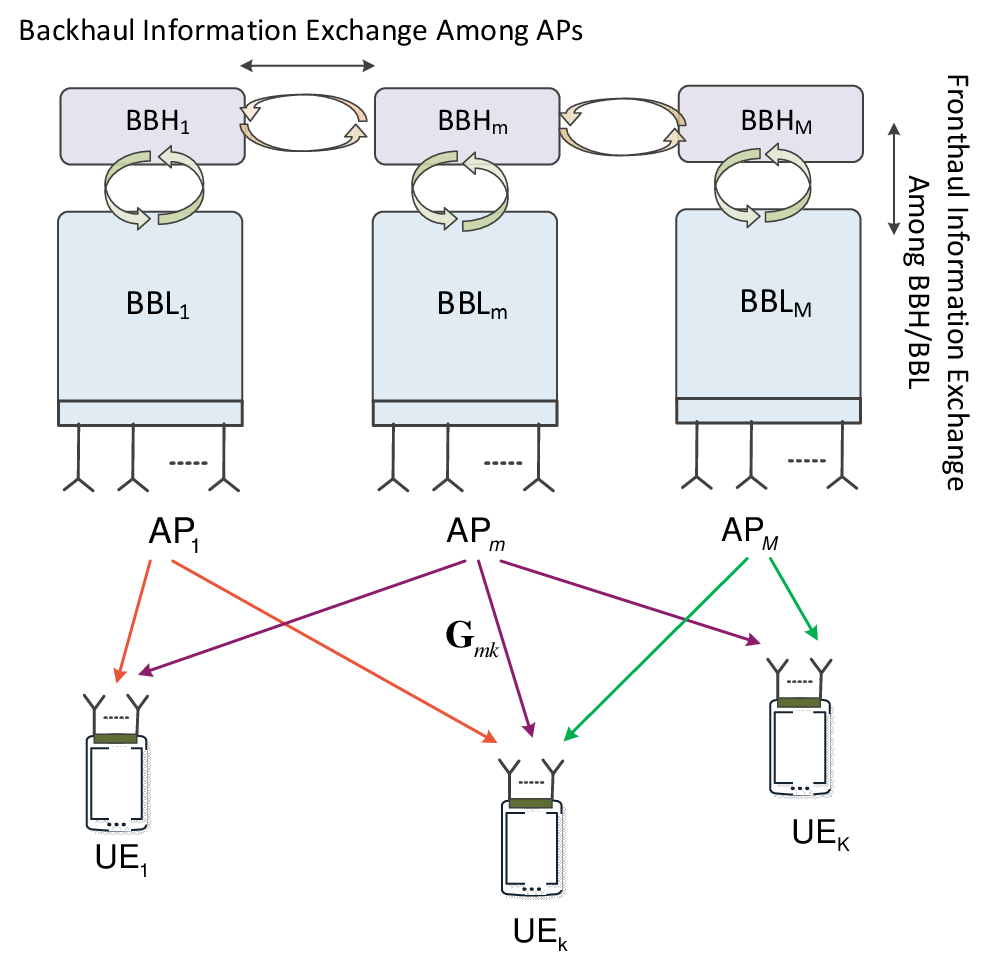}
	\caption{ Illustration of a  CF-mMIMO with BBH and BBL architecture.}
        \vspace{0.5em}
	\label{fig:Fig1}
\end{figure}
\textit{Notation:} We use bold upper case letters to denote matrices, and lower case letters to denote vectors. The superscript   $(\cdot)^H$  stands for the   conjugate-transpose and  tr$(\cdot)$  shows the transpose. A zero-mean circular symmetric complex Gaussian distribution having a variance of $\sigma^2$ is denoted by $\mathcal{CN}(0,\sigma^2)$, while $\mathbf{I}_N$ denotes the $N \times  N$ identity matrix.  Finally, $\mathbb{E}\{\cdot\}$ denotes the statistical expectation.
\section{System Model and Beamforming Design}
We consider a  CF-mMIMO  system  that consists of $M$ APs and $K$  users.  Each AP and user are equipped with $L$ antennas and  $\Nrx$ antennas, respectively. We also consider functional split for baseband unit (BBU). More specifically,  computational functions at BBUs are split into two  entities called baseband low (BBL) and baseband high (BBH), as shown in Fig.~1. The BBH is responsible for managing processing tasks such as beamforming, encoding,  and radio resource management, while the BBL is responsible for processing
tasks such as weight applications, error correction, and modulation.
BBHs are connected through limited-capacity fronthaul links\footnote{In some literature, the term ``backhaul" denotes the link between BBL and BBH when adopting a functional split in BBU. In this paper, we interchangeably use ``fronthaul" to denote the same link, without loss of generality.} to their associated BBLs for sending  information such as  beamforming matrices/vectors, data, and power coefficients. Moreover, BBHs are connected through backhaul links for information exchange between APs. 

We assume a frequency-flat slow fading channel model for each orthogonal frequency-division multiplexing (OFDM) subcarrier. In the sequel, the subcarrier index will be omitted for the sake of notational simplicity.
Let $\qG_{mk}\in \mathbb{C}^{L\times \Nrx}$ be the complex channel matrix between  the $m$-th AP,  and the $k$-th user. It can be modeled as
\begin{equation}
 {\qG}_{mk} = \beta ^{1/2}_{mk}\qH_{mk},
\end{equation}
where $\beta_{mk}$ denotes the large-scale fading coefficient that includes path-loss and
shadowing effects, while  $ {\qH}_{mk}\in \mathbb{C}^{L\times \Nrx}$ is the small-scale fading matrix whose entries are i.i.d. $\mathcal{CN} (0, 1)$ random variables. Large-scale fading coefficients change slowly and may be constant in range of many small-scale fading coherence intervals (over time and frequency bands). Hence, it is assumed that these coefficients are priory known  at each BBL/BBH.

The system is operating according to a time division duplex (TDD) protocol, i.e., the uplink and downlink transmissions occur at different times but use the same frequency. Here, we focus on the downlink  and hence, consider two-phase transmission protocol wherein each coherence interval divides into  uplink training phase and downlink data transmission phase. By exploiting uplink/downlink channel reciprocity, the downlink CSI can be estimated
through uplink training. Let us denote the length of the TDD interval by $\tau_\mathrm{c}$, which is determined by the shortest coherence interval of all users in the network. Also, denote  $\tau_\mathrm{u}$ as the length of the uplink training phase  per coherence interval.

We consider CF-mMIMO with the notion of joint  CB and  UC clustering.  We not that  in the sequel, the set of APs that send signal to user $k$ is called UC cluster of user $k$ and the set of coordinated APs that exchange CSI and perform CB is called CB cluster.
In the downlink data transmission phase, APs perform  CB and send data symbols to a subset  of the users in the system.
In particular, for CB clustering the available  APs are divided into $C$ disjoint CB clusters. APs in each CB cluster share  CSIs such that they have  CSI knowledge of all the users assigned to the APs in this cluster and perform coordinated beamforming.
 Figure 2 illustrates an example of   CF-mMIMO that includes  three users served
by a large number of APs, three UC clusters and two CB clusters. The coloured regions show which UC clusters of AP are serving which users. Clearly, the UC clusters are
partially overlapping which results from the core feature of the cell-free networks. Moreover, the red and brown regions show the  clusters for CSI exchange and CB.

It is noteworthy that the proposed framework for CF-mMIMO with CB and UC clustering is very general and can be degenerated to four special cases as  conventional CF-mMIMO with pure CB,  CF-mMIMO with pure UC clustering, UC CF-mMIMO with fully centralized beamforming and UC CF-mMIMO with fully distributed beamforming as will be discussed in further details in the next sections.
 In the following we will summarize the uplink training and downlink payload transmission phases in more details.
\begin{figure}[t]
	\centering
	\vspace{0em}
	\includegraphics[width=80mm, height=60mm]{  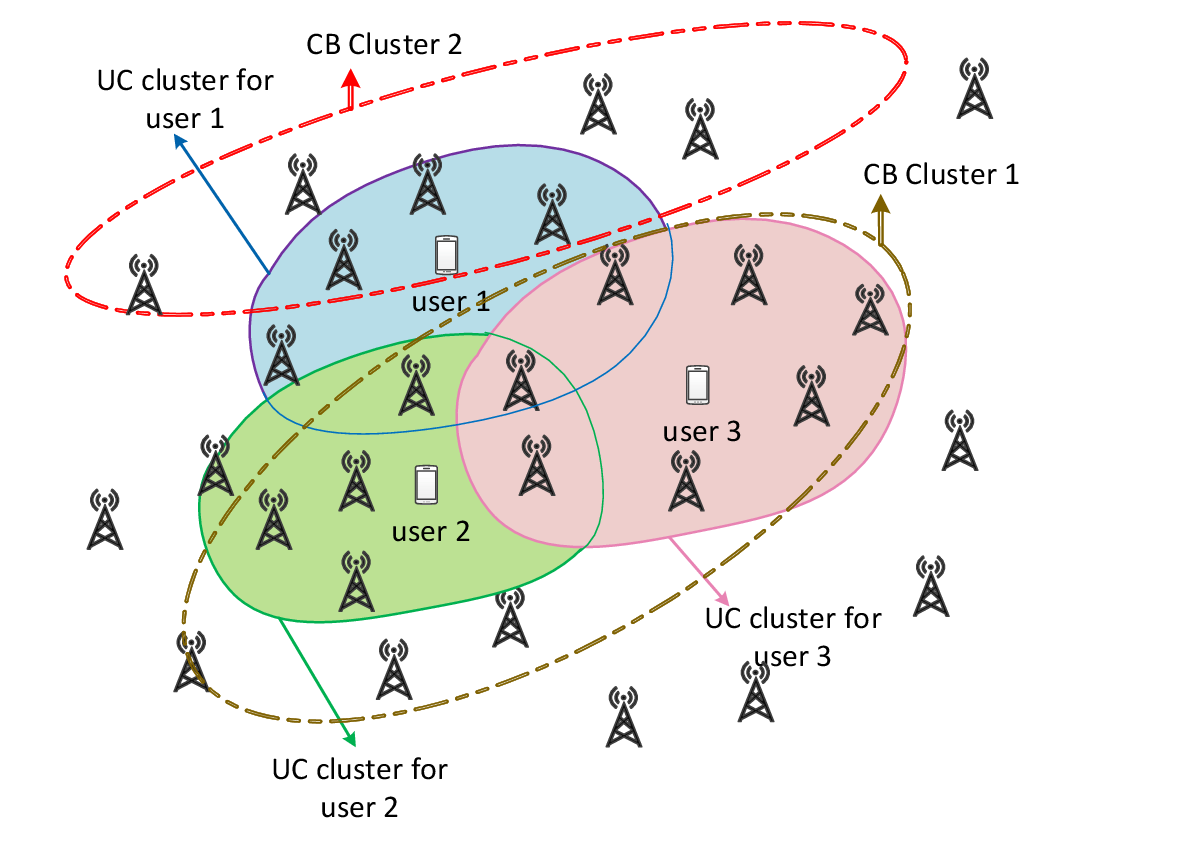}
	\caption{   CF-mMIMO with UC and CB clustering approaches.}
	\label{fig:Fig1}
    \vspace{0.5em}
\end{figure}
\subsection{Uplink Training }
During the uplink training phase, all users send pilot signals to the APs. Accordingly, each AP can estimate the corresponding channels to all the users using the obtained pilot signal. Note that at each AP, the channel estimation is performed at its BBH.
Let us denote by   $\pmb{\Phi}_{\mathrm{u},k}\in \mathbb{C}^{\tau_{\mathrm{u}}\times {\Nrx}}$ the pilot matrix of the $k$-th user where its $n$-th column satisfies $\|\boldsymbol{\phi}_{\mathrm{u},k,n}\|^2=1$, $\forall n \in \Nrx$. The received signal at the $m$-th AP can be written as
\begin{equation}
\qY_{\mathrm{u}, m}=\sum _{k=1}^{K}\sqrt {\tauu}\Pu\qG_{mk}\pmb {\Phi }_{\mathrm{u},k}^{H}+\qW_{\mathrm{u},m},
\end{equation}
where $\Pu$ is the transmit power of each uplink pilot symbol and $\qW_{\mathrm{u},m} \in \mathbb{C}^{L\times \tau_\text{u}}$ is the complex additive Gaussian noise matrix at the $m$-th AP whose elements are i.i.d. $\mathcal{CN}(0,\Snb)$. By projecting $\qY_{\mathrm{u}, m}$ onto $\pmb {\Phi }_{\mathrm{u},k}$, we get
\begin{align}\label{eq:y-mk}
\qY_{\mathrm{u},mk}=&\qY_{\mathrm{u},m}\pmb {\Phi }_{\mathrm{u},k} \nonumber\\
=&\sum _{i=1}^{K}\sqrt {\tauu\Pu}\qG_{mi}\pmb {\Phi }_{\mathrm{u},ik} + \qW_{\mathrm{u},mk},
\end{align}
where $\pmb {\Phi }_{\mathrm{u},ik}=\pmb {\Phi }^{H}_{\text{u},i}\pmb {\Phi }_{\mathrm{u},k}$ and $\qW_{\mathrm{u},mk}=\qW_{\mathrm{u},m}\pmb {\Phi }_{\mathrm{u},k}$.   Now, the estimate of the channel matrix $ \qG_{mk}$ can be calculated as~\cite{Mai:2020:TCOM}
\begin{equation}
 \widehat {\qG}_{mk} = \qY_{\mathrm{u},mk} \qA_{mk},
 \end{equation}
 where
\begin{equation} \label{eq:UL-A1}
\qA_{mk} \triangleq \sqrt {\tauu\Pu}\beta _{mk}\left ({\tauu\Pu\sum _{i=1}^{K}\beta _{mi}\pmb {\Phi }^{H}_{\mathrm{u},ik}\pmb {\Phi }_{\mathrm{u},ik}+ \Snb\qI_{\Nrx}}\right)^{-1}.
\end{equation}

\begin{Remark}
 In the case that pilot sequences assigned to the  users are pairwisely orthonormal,  we have $\pmb {\Phi }_{\mathrm{u},ik}=\qI_{N}$, when $i=k$. Otherwise $\pmb {\Phi }_{\mathrm{u},ik}=\boldsymbol{0}$. Then~\eqref{eq:UL-A1} becomes

\begin{equation}
\qA_{mk} = c_{mk}\qI_{\Nrx}.
\end{equation}
where $c_{mk}\triangleq \frac{\sqrt {\tauu\Pu}\beta _{mk}}{\tauu\Pu \beta _{mk} +  \Snb}.$
Moreover,~\eqref{eq:y-mk} simplifies to
\begin{equation} \label{eq:UL-Ymk}
\qY_{{\mathrm{u}},mk}=\sqrt {\tauu\Pu}\qG_{mk} + \qW_{\mathrm {u},mk},
\end{equation}
and accordingly thermal noise is the only disturbance harming the channel
estimate. A necessary condition for this to happen is $ \tauu\geq KN$.
The estimation error is given by $  \widetilde {\qG}_{mk}= \qG_{mk} -\widehat {\qG}_{mk}$ with $\widehat {\qG}_{mk}$ and $\widetilde {\qG}_{mk}$ are independent and have i.i.d. $\mathcal{CN}(0, \gamma_{mk})$ and i.i.d.  $\mathcal{CN}(0, \beta_{mk}-\gamma_{mk})$ elements, respectively, where $\gamma_{mk}$ is the mean-square of the estimate,
i.e., for any  element $[\widehat {\qG}_{mk}]_{ij}$ we have
$\gamma_{mk}\triangleq \frac{ {\tauu\Pu}\beta _{mk}^2}{\tauu\Pu \beta _{mk} +  \Snb}.
$
\end{Remark}
    \subsection{Downlink Data Transmission and Beamforming Design} \label{sec:Downlink Data Transmission}
In CF-mMIMO with UC clustering, each AP communicates only with a subset of users in the network. The procedure for the selection of the users to serve, i.e., determining the UC clusters, will be specified in the next section.  Denote by $\Km$ the set of users served by the $m$-th AP and define  $K_m=|\Km|$ as the number of users in $\Km$. Now, for all the subsets $\Km$,  $m = 1,\ldots,M$, we can define the set $\Mk$ as the set of all APs that serve the $k$-th user, i.e.,
\begin{equation*}
	 \Mk=\{ m: \, k \in \Km \}, 
 \end{equation*}
with $M_k=|\Mk|$. 
 Let $\Ntx $ be the number of independent downlink data streams sent to each user, $\Ntx \leq N$.
 To be more general and to reduce the amount of data traffic to be fronthauled between BBHs and BBLs, we suppose $\Ntx$ can take an integer value from 1 to $N$. We assume that $\qx_{k}$  is the vector of $\Ntx$ symbols, intended for the $k$-th user, $\mathbb{E}\{\qx_{k}\qx_{k}^H \} = \qI_{\Ntx}$.
Accordingly, the $L \times 1 $ signal transmitted by the $m$-th BBL  can be expressed as
\begin{equation}\label{eq:DL-x1}
	\qs_{m}= \sqrt {\Pb }\sum\nolimits _{k\in \Km}\eta ^{1/2}_{mk} \qQ_{mk}\qx_{k},
\end{equation}
where $\qQ_{mk}$ denotes the downlink beamforming matrix associated with the $k$-th user.
Moreover, $\Pb$ denotes the AP transmit power   and  $\eta_{mk}$  represents  the $k$-th user power control coefficient which needs to satisfy the  power constraint at each AP, i.e.,  $\mathbb{E}\{\|\qs_{m} \|^2 \} \leq P$.
The  $k$-th user receives signal contributions from all the APs; the observable vector is given by
\begin{align}\label{eq:DL-W4}
\textbf {r}_{k}=&\sum _{m=1}^{M}\qG^{H}_{mk}\qs_{m} +  {\qn}_{k}\nonumber\\
&=\sqrt {\Pb }
\sum\nolimits_{m \in \Mk}
\sum\nolimits_{k'\in \Km}
\eta ^{1/2}_{mk'}
\qG^{H}_{mk} \qQ_{mk'}\qx_{k'} +  {\qn}_{k},
\end{align}
where $\qn_k \in \mathbb{C}^{ \Nrx \times 1}$ is the  noise vector whose elements are i.i.d. $\mathcal{CN}(0,\Sn)$.
We can equivalently reformulated \eqref{eq:DL-W4} as
\begin{align}\label{eq:DL-W50}
\textbf {r}_{k}
&=\sqrt {\Pb }
\sum\nolimits _{m\in \Mk}
\eta ^{1/2}_{mk}
\qG^{H}_{mk} \qQ_{mk}\qx_{k}\nonumber\\
&+\sqrt {\Pb }
\sum\nolimits _{\substack {k'=1 \\ k' \neq k}}^{K}
\sum\nolimits _{m\in \Mkp}
\eta ^{1/2}_{mk'}
\qG^{H}_{mk} \qQ_{mk'}\qx_{k'} + \qn_k.
\end{align}
The first term is the desired signal, the second term
describes the multi-user interference (all the signal components intended for user $k', k' \neq k$).

In order to mitigate  interference from other APs, we employ ZF-based CB
which enables a  cluster of APs to share their local CSI and construct  ZF precoder to cancel the   intra-cluster interference.
Let $\Csm$ denotes the set of APs which are in the same CB cluster with AP $m$ and define  $C_m=|\Csm|$ as the number of APs in $\Csm$.  Specifically, if AP $n$  and AP $m$ belong to the same CB cluster then  $\mathcal{C}_n=\Csm$. Let us denote the set $\Usm$ by the union of the sets $\Kmp$  for $m' \in \Csm$, as
\begin{equation*}
	\Usm= \bigcup_{ m' \in \Csm}\Kmp.
\end{equation*}
Clearly, each user in $\Usm$  is served by at least one of the APs in $\Csm$. Since UC clusters are partially overlapping, there might be repeated user indices in $\Usm$. In this case, repeated indices are removed and  $\Usm$ is updated to include only the unique entries of the user indices.

Let
$\widehat {\qG}_{m}$ be an $(L C_m)\times (\Nrx  U_m)$ collective channel estimation matrix for corresponding APs in set  $ \Csm$, where $U_m=|\Usm|$ denotes the number of users in $\Usm$. More specifically, $\widehat {\qG}_{m}$ consists of  $C_m U_m$ blocks of dimension $L\times\Nrx$, $\widehat {\qG}_{ij}$, each corresponding to a particular AP $i$ in set  $ \Csm$ and user $j$ in set  $ \Usm$ as

\begin{equation}\label{eq:Qij}
\widehat {\qG}_{m}=[\widehat {\qG}_{ij}: i \in \Csm, j \in \Usm].
\end{equation}
 Now, by  employing  ZF precoding technique the whole CB matrix constructed for the APs in $\Csm$,  $\overline{\qQ}_{m} \in \mathbb{C}^{(LC_m)\times (\Nrx  U_m)}$,  can be written as
\begin{equation}\label{eq:Q}
\overline{\qQ}_{m}=
\widehat {\qG}_{m} \left ({\widehat {\qG}_{m}^{H} \widehat {\qG}_{m} }\right)^{-1}.
\end{equation}
The CB matrix for APs in set $\Csm$ can be rewritten as $\overline{\qQ}^{\mathrm{ZF}}_{m}= \overline{\qQ}_{m} \overline{\qP}$, where $\overline{\qP}$   includes the power allocation coefficients. To avoid interference toward the users in $\Usm$,  $\widehat {\qG}_{m}^H\overline{\qQ}^{\mathrm{ZF}}_{m}$ should be diagonal~\cite{nayebi:2017:wcom} and hence it is necessary to have $\eta_{m'k}= \eta_{k,\cm}$, $\forall m' \in \Csm$. That is, in a given CB cluster, power coefficients are only functions of $k$
and $\eta_{k,\cm}$ shows the power coefficient allocated for all APs in set $\Csm$ to user $k$. Thus, the beamforming matrix can be further expressed
as
\begin{align}
\overline{\qQ}^{\mathrm{ZF}}_{m}=\overline{\qQ}_{m}\mathbf{P}_m,
\end{align}
where $\mathbf{P}_m$ is an $NU_m \times NU_m$ block diagonal matrix with the $j$-th  matrix on the diagonal ${\eta^{1/2}_{k,\cm}} \qI_{\Nrx}$ and $j$ is the index of user $k$ in set $\Usm$. Here, we also note that to implement ZF-based precoding, the total number of  transmit antennas at   APs in CB $\cm$, must meet the requirement $L   C_m \geq \Nrx  U_m$.

Since $\Ntx \leq N$, we employ singular value decomposition (SVD) technique which enables each BBH  to construct its precoding matrix from the whole beamforming matrix designed for each CB cluster. Performing the SVD of  $ \widehat {\qG}_{m}$  yields
\begin{equation}\label{eq:G}
\widehat {\qG}_{m} = \qU_m \boldsymbol{\Sigma}_m \qV^H_m,
\end{equation}
where the matrices $\qU_m \in \mathbb{C}^{(LC_m) \times (NU_m)}$ and $\qV_m \in \mathbb{C}^{(NU_m) \times (NU_m)}$ contain,  the left and right singular vectors and $\boldsymbol{\Sigma}_m$ includes singular values in nonascending order, corresponding to the
non-zero eigenmodes, respectively.
Substituting \eqref{eq:G} into \eqref{eq:Q}, we have
\begin{align}\label{eq:Q2}
\overline{\qQ}_{m}&=\qU_m \boldsymbol{\Sigma}_m \qV^H_m \left (\qV_m \boldsymbol{\Sigma}_m \qU^H_m  \qU_m \boldsymbol{\Sigma}_m \qV^H_m\right)^{-1}\nonumber\\
&=\qU_m \boldsymbol{\Sigma}_m \qV^H_m \left (\qV_m \boldsymbol{\Sigma}^2_m \qV^H_m\right)^{-1}\nonumber\\
&=\qU_m \boldsymbol{\Sigma}^{-1}_m \qV^H_m.
\end{align}
Now, each BBH can  construct its precoding matrix for transmission to  user $k$ by choosing the corresponding $\Ntx$ column vectors associated to its strongest singular values as expressed in~\eqref{eq:Qmk}.
In particular, let AP $m$ correspond to the  $i$-th element of set $\Csm$,  $i\in \{1,\cdots,C_m\}$, and user $k$ correspond to the $j$-th element of set $\Usm$,  $j\in \{1,\cdots,U_m\}$,
 then we set $\omega_{mk}=i$ and $\nu_{mk}=j$, respectively. Accordingly, the downlink beamforming matrix constructed for the $m$-th AP to the $k$-th user can be calculated as the submatrix of $\qQ_{m}$ obtained by selecting the $L$ rows $a_{mk}$ to $\check{a}_{mk}$ and $\Ntx$ columns $b_{mk}$ to $\check{b}_{mk}$ as
\begin{align}\label{eq:Qmk}
\qQ_{mk}&= \left[{\overline{\qQ}_{m}}\right]_{(a_{mk}: \check{a}_{mk}, b_{mk}:\check{b}_{mk})}\nonumber\\
&=\qU_{mk}  \boldsymbol{\Sigma}_m^{-1} \qV^H_{mk},
\end{align}
where $\qU_{mk}\triangleq[\qU_m ]_{(a_{mk}: \check{a}_{mk},:)}$  with $a_{mk}=(\omega_{mk}-1)\times L+1$ and $\check{a}_{mk}=a_{mk}+ L-1$. Also, $\qV^H_{mk}\triangleq [\qV^H_m]_{(:, b_{mk}:\check{b}_{mk})}$ with $b_{mk}=(\nu_{mk}-1)\times\Nrx+1$ and  $\check{b}_{mk}=b_{mk}+\Ntx-1$.

We note that, in the proposed low complexity ZF-based CB all the channels ${\qG}_{ij}$ between AP $i$, $i \in \Csm$, and user $j$, $j \in \Usm$, are accounted in CB design of cluster $\Csm$. However, according to the UC clustering, each AP $i$  may not transmit  to  one or some of the users in  $\Usm$ or equivalently  each user $j$ may  not be served  by  one or some of the APs in  $\Csm$. This might affect the orthogonality in  beamforming design  and result intra-cluster interference. We argue that, in real life  CF-mMIMO scenario with efficient UC clustering, the portion of this intra-cluster interference is insignificant compared to the desired signal.
The main reason is that  those APs that are far from a given user $j$ and have  smaller channel gains are not selected for UC cluster of user $j$, which can be interpreted as the non-orthogonality and intra-cluster interference due to the proposed CB design  is negligible and does not degrade the performance. 
We will justify this assumption in Section~\ref{sec:Numerical Results}.

\section{Performance Analysis and Fronthaul Requirements}
In this section, we first present analytical expression for the SE of the   CF-mMIMO  with CB  assuming that each user uses   MMSE-SIC scheme to detect the desired symbols. Moreover, we provide a practical formulation for the fronthaul requirement of CF-mMIMO with BBH-BBL operation. 
\subsection{Performance Analysis}
The received signal at the $k$-th user can be rewritten as
\begin{align}\label{eq:DL-W5}
\textbf {r}_{k}
=
\sqrt {\Pb }\sum _{k'=1}^K
{\qD}_{kk'} {\qx}_{k'}
 + \qn_k,
\end{align}
where ${\qD}_{kk'}=\sum _{m\in \Mkp}\eta_{k', {\cm}}^{1/2}\qG^{H}_{mk}\qQ_{mk'}$.
Accordingly, we can derive an achievable downlink SE of user $k$ as follows.
\begin{proposition}
\label{prop:general_rate_exact_orthogonal}
For general   CF-mMIMO with ZF-based CB  given statistic CSI\footnote{In practical scenarios, channel statistics can be acquired through various methods, including time-averaged estimation, user feedback, model-based predictions, and pilot signals. Moreover, machine learning   techniques are increasingly being used to acquire channel statistics from historical data and user context~\cite{Wang:2019:WCL}.}  and using MMSE-SIC detectors, an achievable downlink SE of user $k$ can be calculated as
\begin{equation}\label{eq:rate_general_o}
 R_{k} = (1 \!-\! \tauu/\tau_\mathrm{c}) \log _{2} \left |{ \qI_{\Ntx} \!+ \! \Pb \bar {{\qD}}^{H}_{kk} \left ({\boldsymbol{\Psi }^{b}_{kk}}\right)^{-1} \bar {{\qD}}_{kk} }\right |,
 \end{equation}
where
\begin{align}
\bar {{\qD}}_{kk}= \sum _{m\in \Mk}\eta ^{1/2}_{k, \cm} \mathop {\mathrm {\mathbb {E}}}\nolimits \left \{{ {\qG}^{H}_{mk} } \qQ_{mk}\right\},
\end{align}
and
\begin{align*}
{\boldsymbol{\Psi }^{b}_{kk}}=&\Pb\mathbb {E} \bigg\{\sum _{k'=1}^{K}     {{\qD}}_{kk'}  {{\qD}}_{kk'}^{H}\bigg\}-
\Pb\bar {{\qD}}_{kk}\bar {{\qD}}^{H}_{kk} + \Sn\qI_{N}.
\end{align*}
\end{proposition}
\begin{proof}
The achievable downlink SE of the $k$-th user, $k\in \Km$,
with MMSE-SIC detection scheme for the received signal $\textbf{r}_k$ in \eqref{eq:DL-W5} and side information $\boldsymbol{\Theta }_{k}$ (assuming that $\boldsymbol{\Theta }_{k}$ is independent of $\qx_k$) is given by~\cite{Mai:2020:TCOM}
\begin{equation}\label{eq:DL-W6}
R_{k} = (1 - \tau _{\text {tot}}/\tau _{\text {c}})\mathbb {E}\left \{{ \log _{2} \left |{ \qI_{\Ntx} + \boldsymbol{\Upsilon }^{a}_{kk} }\right |}\right \},
\end{equation}
where
\begin{align}
 \boldsymbol{\Upsilon }^{a}_{kk}\triangleq\Pb\mathbb {E}\{{\qD}^{H}_{kk}|\boldsymbol{\Theta }_{k}\} ({\boldsymbol{\Psi }^{a}_{kk}})^{-1} \mathbb {E}\{{\qD}_{kk}|\boldsymbol{\Theta }_{k}\},  
  \end{align}
  and
  \begin{align}
  {\boldsymbol{\Psi }^{a}_{kk}}\triangleq &\qI_{N} +\mathbb {E}\big\{(\Pb\sum_{k'=1}^K{\qD}_{kk'}{\qD}^{H}_{kk'}|\boldsymbol{\Theta }_{k})\big\}-\Pb\mathbb {E}\{{\qD}_{kk}|\boldsymbol{\Theta }_{k}\}\nonumber\\
  &\times\mathbb {E}\{{\qD}^{H}_{kk}|\boldsymbol{\Theta }_{k}\}.
 \end{align} 
with $\tau _{\text {tot}}$ denotes the total training duration for each coherence interval.
Given statistic CSI, we have $\boldsymbol{\Theta }_{k}=\bar {{\qD}}_{kk}=\mathbb {E} \{{\qD}_{kk}\}$  and  $\tau _{\text {tot}}=\tauu$ and hence, using~\eqref{eq:DL-W6},  we can obtain~\eqref{eq:rate_general_o}.
\end{proof}

Note that to obtain \eqref{eq:rate_general_o}, user $k$ needs to know only the statistical properties of the channels. Therefore, there is no need for downlink pilots, which can reduce the resources required for downlink channel estimation, such as transmit power and estimation overhead. This achievable SE corresponds the rate obtained from the hardening bouding technique which is widely used in massive MIMO literature.

\begin{Remark}   
The SE analysis in Proposition~\ref{prop:general_rate_exact_orthogonal} is  general and holds for different UC and CB clustering schemes. When each user is served by a set of APs, $M_k\leq M$ and $C_m =1$, $\forall m$, we have CF-mMIMO with pure UC approach. In contrast, when each user served by all the APs, i.e., $M_k=M$, $\forall k$, and $C_m > 1$, we have CF-mMIMO with pure CB approach.
Further, the SE expression for  CF-mMIMO can be easily computed for any $\Csm$ and the corresponding $\Usm$  set. Two important special cases of CB clustering are described below and compared with each other in Section~\ref{sec:Numerical Results} by means of numerical results.

{1)   CF-mMIMO with Fully Centralized Beamforming:}
In the case of CF-mMIMO with fully centralized beamforming which is a form of network MIMO  all the CSIs  are
exchanged by all the APs (through the backhaul links between BBHs) to cancel interference. More specifically,   there is one CB cluster consists of all APs in the network or equivalently $\Csm=\{1,\cdots,M\}, \forall m$, $\Usm=\mathcal{U}=\{1,\cdots,K\}, \forall m$,    and $\eta_{k,\cm}= \eta_k$.


{2)  CF-mMIMO with Fully Distributed Beamforming:}
In this case, the beamforming can be done locally at each AP without CSI exchange with other APs. This makes the network truly distributed. In particular, each BBH at each AP utilizes its own local channel
estimates to generate a ZF precoder.   CF-mMIMO with distributed beamforming corresponds to the case when there are $M$ beamforming clusters, i.e., $C=M$, $\Csm=m$, $\forall m$, with cardinality $1$, each  includes $\Usm=\Km$ users. Also, in this case power control coefficient for user $k$ is  $\eta_{k, \cm}=\eta_{k, m}$.
In this case, each AP can only mitigate its own interference, but not interference from other APs.
Note that this  can be  seen as a particular CF-mMIMO configuration with pure UC clustering and fully distributed beamforming which does not require any CSI to be exchanged between BBHs.

\end{Remark}

\subsection{Fronthaul Requirements}
\label{app:Fronthaul Requirements}
In a CF-mMIMO system with BBH-BBL operation, precoding weights are calculated at BBH. Then precoding weights and user data are sent to BBL via fronthaul links. 
We consider packed-based evolved CPRI (eCPRI) for the fronthaul  transmission. More specifically, fronthaul requirement of $m$-th AP for downlink data transmission in the considered CF-mMIMO system with BBH-BBL operation can be calculated as
\begin{align}\label{eq:FH1}
\mathrm{FH}_{m,\mathrm{data}}= \frac{\log_2 (M_{\mathrm{order}})\Ntx K_m N_{\mathrm{subcarrier}} N_{\mathrm{OFDM}}}{\mathrm{eCPR_{eff}}\mathrm{delay_{data}}},
\end{align}
where $M_{\mathrm{order}}$ is the modulation order, $N_{\mathrm{subcarrier}}$ denotes the number of OFDM subcarrier, and $\mathrm{delay}_{\mathrm{data}}$  and $\mathrm{eCPR}_{\mathrm{eff}}$ are constant parameters.
Moreover, the required fronthaul capacity for beamforming weights is given by
\begin{align}\label{eq:FH2}
\mathrm{FH}_{m,\mathrm{BF}}= \frac{2L \Ntx K_m N_{\mathrm{bits}}N_{\mathrm{Gran}}}{\mathrm{eCPR_{eff}}\mathrm{delay_{precoder}}} ,
\end{align}
where $N_{\mathrm{bits}}$ is the number of quantization bits and $N_{\mathrm{Gran}}$ is the beamforming granularity. More specifically, if we denote  the number of subcarriers grouped under one precoding operation as $N_{\mathrm{Group}}$, then $N_{\mathrm{Gran}}=\frac{N_{\mathrm{subcarrier}}}{N_{\mathrm{Group}}}$.
Based on~\eqref{eq:FH1} and~\eqref{eq:FH2}, fronthaul requirement scales with the number of users that each AP serves and also the number of data streams to be sent to each user. Accordingly, in the proposed  CF-mMIMO system, we constrain the downlink traffic bandwidth to be fronthauled  by limiting the number of users that each AP  serves  and the number of independent downlink data streams sent to each user, i.e.,
\begin{itemize}
	\item Each BBH $m$ only needs to send  data and precoding matrix of $K_m \leq K$ users to its BBL through fronthaul link (this is in contrast with the conventional CF-mMIMO system, which requires    fronthauled information of all $K$ users).
	\item  We suppose $\Ntx$  takes an integer value from $1$ to $N$. Therefore, each fronthaul link needs to support  $\Ntx \leq N$ downlink data streams per user.
\end{itemize}
On the other hand, a higher modulation order, $M_{\mathrm{order}}$, is essential to achieve the high data rates. However, from~\eqref{eq:FH1}, fronthaul requirement for data transmission increases with $M_{\mathrm{order}}$, thus, there is a trade-off  between the SE and fronthaul requirement.
From the information-theoretic perspective,  modulation-constrained achievable information rate, $\mathcal{R}_{\mathrm{mo}}$, is bounded as $\mathcal{R}_{\mathrm{mo}}< \bar{\mathcal{C}}$, where $\bar{\mathcal{C}}$ is  the AWGN channel capacity. Besides, $\mathcal{R}_{\mathrm{mo}}$ cannot be greater than the entropy of the corresponding constellation, i.e., $\mathcal{R}_{\mathrm{mo}}< \log_2 (M_{\mathrm{order}})$. Thus,
$\mathcal{R}_{\mathrm{mo}}$ can be upper-bounded as
 \begin{align}\label{eq:FH3}
\mathcal{R}_{\mathrm{mo}} \leq \mathrm{min} (\bar{\mathcal{C}},\log_2 (M_{\mathrm{order}})).
\end{align}
At high signal-to-noise ratio (SNR),  $\mathcal{R}_{\mathrm{mo}}$ approaches the
constellation entropy $\log_2 (M_{\mathrm{order}})$ and is independent of the SNR, so that the upper bound in~\eqref{eq:FH3} is tight  at high SNR~\cite{Urlea:2021:JLT}.

Accordingly, from~\eqref{eq:FH1},~\eqref{eq:FH2}, and~\eqref{eq:FH3},  the  per-AP fronthaul capacity constraints can be written as
\begin{align}\label{eq:Kmax_cons}
&\mathrm{FH}_{m,\mathrm{data}}+\mathrm{FH}_{m,\mathrm{BF}} =\nonumber\\
&\Ntx K_m\left(\alpha_1 \log_2 (M_{\mathrm{order}}) + \alpha_2\right)\leq \mathrm{FH}_{\mathrm{max}},~~~ \forall m,
\end{align}
and
\begin{align}\label{eq:FHMorder}
 R_{k} \leq  \log_2 (M_{\mathrm{order}}) ,~~~ \forall m,
\end{align}
where $\mathrm{FH}_{\mathrm{max}}$ is the per-AP maximum fronthaul constraint  for downlink data transmission and beamforming weights, $\alpha_1\triangleq  \frac{N_{\mathrm{subcarrier}} N_{\mathrm{OFDM}}}{\mathrm{eCPR_{eff}}\mathrm{delay_{data}}}  $, and $\alpha_2 \triangleq \frac{2 L N_{\mathrm{bits}}N_{\mathrm{Gran}} }{\mathrm{eCPR_{eff}}\mathrm{delay_{precoder}}} $.


\section{ Joint User Association and Power Allocation}
We aim to  optimize UC clustering and  downlink transmit powers for the maximization of the system sum SE, subject to per-AP fronthaul capacity and maximum  transmit power  constraints.  We would like to highlight that the presence of multiple antennas at users in the CF-mMIMO system plays a critical role in the optimization process. In particular, multi-antenna users enhance SINR and detection through spatial diversity and multiplexing, influencing resource allocation via the effective channel gain matrix. On the other hand, multi-antenna configurations increase fronthaul usage, tightening constraints on the number of users served. Our optimization framework accounts for these interactions by incorporating fronthaul-constrained rate limits, ensuring efficient resource allocation under practical constraints.
Accordingly,
in what follows we formulate  joint  user association and power allocation for  a given large-scale fading coherence time.  
\vspace{-1em}
\subsection{Sum-SE Maximization}
With~\eqref{eq:DL-x1} the expected transmit signal power from AP $m$ is given by
\begin{align}
\mathrm {\mathbb {E}}\{\qs_{m}^H \qs_{m}\}
&=\Pb\sum\nolimits_{k \in \Km}\eta_{k,\cm}\text {tr}\left( \mathrm {\mathbb {E}}\{ \qQ_{mk}\qQ_{mk}^H\} \right).
\end{align}
Therefore, the per-AP power constraint can be written as
\begin{align}\label{eq:P_const}
\sum\nolimits_{k \in \Km}\eta_{k,\cm}\text {tr}\left( \mathrm {\mathbb {E}}\{ \qQ_{mk}\qQ_{mk}^H\} \right)\leq 1.
\end{align}

Now, let $\boldsymbol {\eta }$ denote  the set of all power control coefficients $\boldsymbol {\eta } \triangleq \{\eta_{k,\cm}: k=1,\cdots, K; m=1, \cdots, M\}$.
Recall that for the given AP $n$  and AP $m$ in the same CB cluster, we have  $\mathcal{C}_n=\Csm$ and $\eta_{k,\cm}=\eta_{k,\cn} \forall k$.
For convenience, let  $\boldsymbol{\mathcal {K}} \triangleq \{\bigcup\Km, m = 1, \cdots, M\}$  denote the UC clustering control
variable. Now, for the given CB clustering, the  sum-SE maximization problem  can be formulated as
\begin{subequations}\label{eq:problem2}
\begin{align}
&\max _{\boldsymbol {\eta}, \Kset} \sum_{k=1}^K R_{k}(\boldsymbol {\eta },\Kset) \label{eq:problem2:O1}\\
&\text {st.}~\eta _{k,\cm}\geq 0, \qquad  \forall m, k, \label{eq:problem2:C1} \\
&\hphantom {\text {st.}}  R_{k}(\boldsymbol {\eta },\Kset) \leq \log_2 (M_{\mathrm{order}}), \qquad \forall k, \label{eq:problem2:C2} \\
&\hphantom {\text {st.}}\Ntx K_m \left(\alpha_1 \log_2 (\Morder) + \alpha_2\right)\leq \mathrm{FH}_{\mathrm{max}}, \quad \forall m, \label{eq:problem2:C3} \\
&\hphantom {\text {st.}} \sum\nolimits_{k \in \Km}\eta_{k,\cm}\text {tr}\left( \mathrm {\mathbb {E}}\{ \qQ_{mk}\qQ_{mk}^H\} \right)\leq 1, \qquad  \forall m. \label{eq:problem2:C4}
\end{align}
\end{subequations}

Optimization problem~\eqref{eq:problem2} is  obviously nonconvex due to \emph{i)} the presence of unknown optimization
variables $\Km$ and $\eta_{k, \cm}$, $m \in \{1,\cdots, M\},$ $k \in \{1,\cdots,K\}$, appearing as products   in  $R_{k}(\boldsymbol {\eta },\Kset)$ \emph{ii)} beamforming matrix $\qQ$ is a non-convex function of $\Kset$.  Therefore, the
joint optimization of $\boldsymbol {\eta}$ and $\Kset$ cannot be implemented
in a complexity-efficient way  and hence, we separate the optimization variables and solve
the corresponding sub-problems with time-efficient algorithms.

\subsection{Power Allocation Optimization Given  User Association}
The power allocation problem with fixed UC association $\Kset$ is expressed as
\begin{subequations}\label{eq:problem_PA1}
\begin{align}
&\max _{\boldsymbol {\eta}} \sum_{k=1}^K R _{k}(\boldsymbol {\eta }) \label{eq:problem_PA1:O1}\\
&\text {st.}~\eta _{k,\cm}\geq 0,\qquad  \forall m, k, \label{eq:problem_PA1:c1}  \\
&\hphantom {\text {st.}}  R _{k}(\boldsymbol {\eta }) \leq  \log_2 (M_{\mathrm{order}}), \qquad \forall k,   \label{eq:problem_PA1:c2}\\
&\hphantom {\text {st.}} \sum\nolimits_{k \in \Km}\eta_{k,\cm}\text {tr}\left( \mathrm {\mathbb {E}}\{ \qQ_{mk}\qQ_{mk}^H\} \right)\leq 1, \qquad  \forall m. \label{eq:problem_PA1:c3}
\end{align}
\end{subequations}
This problem is  non-convex due to the non-convex objective function~\eqref{eq:problem_PA1:O1} and non-convex constraint~\eqref{eq:problem_PA1:c2}, which
makes its solution difficult. Thus, we resort to sequential optimization framework, also known as 
MM approach. By using this iterative framework, we can reformulate the original
non-convex problem into a series of auxiliary convex subproblems that
can be easily solved by algorithms that have  low
computational complexity and their solutions converge to
a locally optimal solution of the original problem\cite{Hunter:2004:MM}. 
The proposed MM strategy, at each iteration, solves an auxiliary problem whose objective function is a lower-bound of Problem~\eqref{eq:problem_PA1:O1}.
Such a lower bound must fulfill the following criteria: it  must be equal to the original function~\eqref{eq:problem_PA1:O1}; its gradient and of the original objective function~\eqref{eq:problem_PA1} must be equal at the point $\boldsymbol {\eta }_k$, which is updated at each algorithm's iteration.
 To this end, we need to derive a proper lower-bound of the objective function of~\eqref{eq:problem_PA1}. Using the fact that $\mathrm{det}(\qI_{M}+\qA_1\qA_2)=\mathrm{det}(\qI_{N}+\qA_2\qA_1)$, where $\qA_1$ and $\qA_2$ are matrices of size $M\times N$ and $N\times M$, respectively, the SE of the user $k$, $k=1, \cdots, K$,  can be rewritten as
\begin{align}\label{eq:R_PA}
R_{k}  (\boldsymbol {\tieta }) &=(1 \!-\! \tauu/\tau_\mathrm{c})\left[g\left(\boldsymbol {\Omega}_k  (\boldsymbol {\tieta })\right)-g\left(\boldsymbol {\Xi}_k  (\boldsymbol {\tieta })\right)\right],
 \end{align}
where $\tieta\triangleq\{\etkpmn: k'=1,\cdots, K; m=1, \cdots, M, n=1, \cdots, M\}$, $\etkpmn=\eta ^{1/2}_{k',\cm}\eta ^{1/2}_{k', \cn}$,   $g(\cdot)\triangleq\log_2|(\cdot)|$,
\begin{align} \label{eq:omega}
\boldsymbol {\Omega}_k   ( {\tieta })&=\Pb\sum _{k'=1}^{K} \sum _{m \in \Mkp}\sum _{n \in \Mkp} \etkpmn\mathbb {E}\big\{\qG^{H}_{mk}{\qQ}_{mk'} \times \nonumber\\
&\qQ^{H}_{nk'}\qG_{nk}\big\}+ \Sn\qI_{N},
\end{align}
and
\begin{align} \label{eq:si}
\boldsymbol {\Xi}_k   ({\tieta })&=\Pb\sum _{k'=1}^{K} \sum _{m \in \Mkp}\sum _{n \in \Mkp} \etkpmn\mathbb {E}\{\qG^{H}_{mk}{\qQ}_{mk'} \times\nonumber \\
&\qQ^{H}_{nk'}\qG_{nk}\} -\Pb\sum _{m \in \Mk}\sum _{n \in \Mk}\etkmn  \mathop {\mathrm {\mathbb {E}}}\nolimits  \{{ \qG^{H}_{mk} } \qQ_{mk}\} \times\nonumber\\
&\mathop {\mathrm {\mathbb {E}}}\nolimits \{\qQ^{H}_{nk}{ \qG_{nk} }\}+ \Sn\qI_{N}.
\end{align}
Each term of the sum in both $\boldsymbol {\Omega}_k  ( {\tieta })$ and $\boldsymbol {\Xi}_k  ( {\tieta })$ in~\eqref{eq:omega} and~\eqref{eq:si}, respectively, is concave. In addition,  the summation preserves concavity and the function $\log_2|(\cdot)|$
is matrix-increasing. 
As a result,  $g(\boldsymbol {\Omega}_k  (\boldsymbol {\tieta }))$ and $g(\boldsymbol {\Xi}_k  (\boldsymbol {\tieta }))$ in~\eqref{eq:R_PA} are concave functions.
Problem~\eqref{eq:problem_PA1} is equivalent to
\begin{subequations}\label{eq:problem_PA2}
\begin{align}
&\min _{\boldsymbol {\eta,\tieta}}\quad -\sum_{k=1}^K  {R}_{k}  ( {\tieta }) \label{eq:problem_PA2:O1}\\
&\text {st.}~\eta _{k,\cm}\geq 0,\qquad  \forall m, k, \label{eq:problem_PA2:c1}  \\
&\hphantom {\text {st.}}   {R}_{k}  (\tieta)  \leq  \log_2 (\Morder), \qquad  \forall k,  \label{eq:problem_PA2:c2}\\
&\hphantom {\text {st.}} \sum\nolimits_{k \in \Km}\eta_{k,\cm}\text {tr}\left( \mathrm {\mathbb {E}}\{ \qQ_{mk}\qQ_{mk}^H\} \right)\leq 1, \qquad  \forall m,  \label{eq:problem_PA2:c3}\\
&\hphantom {\text {st.}}\etkpmn^2=\eta_{k',\cm}\eta_{k', \cn}, \quad  \forall k',  m,  n. \label{eq:problem_PA2:c4}
\end{align}
\end{subequations}
We  replace constraint~\eqref{eq:problem_PA2:c4} by
\begin{align}\label{eq:etkpmn1}
\etkpmn^2 \leq \eta_{k',\cm}\eta_{k', \cn},
 \end{align}
\begin{align}\label{eq:etkpmn2}
\etkpmn^2 \geq \eta_{k',\cm}\eta_{k', \cn}.
 \end{align}
From~\eqref{eq:etkpmn1}, it is true that~\eqref{eq:etkpmn2} is equivalent to
\begin{align}\label{eq:S1}
    T(\boldsymbol {\eta}, \tieta) = \sum_{k'=1}^K \sum_{m\in \Mkp} \sum_{n\in\Mkp} (\eta_{k',\cm}\eta_{k', \cn}- \etkpmn^2) \leq 0.
 \end{align}
Accordingly, problem~\eqref{eq:problem_PA2} can be written in a more tractable form
as
\begin{subequations}\label{eq:problem_PA3}
\begin{align}
&\min _{\boldsymbol {\eta,\tieta}}\quad -\sum_{k=1}^K  {R}_{k}  ( {\tieta }) \label{eq:problem_PA3:O1}\\
&\text {st.}~\eta _{k,\cm}\geq 0,\qquad  \forall m, k, \label{eq:problem_PA3:c1}  \\
&\hphantom {\text {st.}}  {R}_{k}  (\tieta)  \leq  \log_2 (\Morder), \qquad  \forall k,  \label{eq:problem_PA3:c2}\\
&\hphantom {\text {st.}} \sum\nolimits_{k \in \Km}\eta_{k,\cm}\text {tr}\left( \mathrm {\mathbb {E}}\{ \qQ_{mk}\qQ_{mk}^H\} \right)\leq 1, \qquad  \forall m,  \label{eq:problem_PA3:c3}\\
&\hphantom {\text {st.}} \etkpmn^2 \leq \eta_{k',\cm}\eta_{k', \cn},  \qquad \forall k', m,  n,\label{eq:problem_PA3:c4}\\
&\hphantom {\text {st.}} T(\boldsymbol {\eta}, \tieta)\leq 0\label{eq:problem_PA3:c5}.
\end{align}
\end{subequations}
Now,  consider the following optimization problem
\begin{subequations}\label{eq:problem_PA4}
\begin{align}
&\min _{\boldsymbol {\eta,\tieta}}\quad \mathcal{L}  (\boldsymbol {\eta},{\tieta }) \label{eq:problem_PA4:O1}\\
&\text {st.}~\eta _{k,\cm}\geq 0,\qquad  \forall m, k, \label{eq:problem_PA4:c1}  \\
&\hphantom {\text {st.}}  {R}_{k}  (\tieta)  \leq  \log_2 (\Morder) \qquad  \forall k,  \label{eq:problem_PA4:c2}\\
&\hphantom {\text {st.}} \sum\nolimits_{k \in \Km}\eta_{k,\cm}\text {tr}\left( \mathrm {\mathbb {E}}\{ \qQ_{mk}\qQ_{mk}^H\} \right)\leq 1, \qquad  \forall m,  \label{eq:problem_PA4:c3}\\
&\hphantom {\text {st.}} \etkpmn^2 \leq \eta_{k',\cm}\eta_{k', \cn}, \qquad \forall k', m, n, \label{eq:problem_PA4:c4}
\end{align}
\end{subequations}
where $\mathcal{L}  (\boldsymbol {\eta},{\tieta })\triangleq-\sum_{k=1}^K  {R}_{k}  ( {\tieta })+ \lambda T(\boldsymbol {\eta},{\tieta })$ is the Lagrangian of \eqref{eq:problem_PA3} and $\lambda$ is the Lagrangian multiplier corresponding to constraint \eqref{eq:problem_PA3:c5}.
%
\begin{proposition}
\label{proposition-dual}
The value $T_{\lambda}$ of $T$ at the solution of \eqref{eq:problem_PA4} corresponding to $\lambda$ converges to $0$ as $\lambda \rightarrow +\infty$. Moreover, problem \eqref{eq:problem_PA3} has strong duality, such that
\vspace{-0.5em}
\begin{equation}\label{Strong:Dualitly:hold}
\underset{\{\boldsymbol{\eta},\tieta\}\in\mathcal{F}}{\min}\,\,
-\sum_{k=1}^K  {R}_{k}  ( {\tieta })
=
\underset{\lambda\geq0}{\sup}\,\,
\underset{\{\boldsymbol{\eta},\tieta\}\in\widetilde{\mathcal{F}}}{\min}\,\,
\mathcal{L}  ({\boldsymbol{\eta},\tieta}),
\end{equation}
where ${\mathcal{F}}$ and $\widetilde{\mathcal{F}}$  are  feasible sets defined as  $\mathcal{F}\triangleq \{\eqref{eq:problem_PA3:c1}, \eqref{eq:problem_PA3:c2}, \eqref{eq:problem_PA3:c3}, \eqref{eq:problem_PA3:c4}, \eqref{eq:problem_PA3:c5}\}$  and $\widetilde{\mathcal{F}}\triangleq \mathcal{F} \setminus {\eqref{eq:problem_PA3:c5}}$, respectively.
Therefore,  at the optimum solution $\lambda^* \geq0$ of~\eqref{Strong:Dualitly:hold}, problems \eqref{eq:problem_PA3} and \eqref{eq:problem_PA4} are equivalent.
\end{proposition}

\begin{proof}
    The proof follows \cite{vu18TCOM}.
\end{proof}
Attaining the optimal solution for problem \eqref{eq:problem_PA3} necessitates $T_{\lambda}$ to be zero, which occurs   when $\lambda\to+\infty$, as indicated by Proposition~\ref{proposition-dual}. However, for practice implementation, it is sufficient to consider $T_{\lambda}\leq\varepsilon$, for some small value of $\varepsilon$ with a sufficiently large value of $\lambda$~\cite{vu18TCOM,Mohammadali:JSAC:2023}.

We note that the objective function~\eqref{eq:problem_PA4} is still non-convex and difficult to handle. Therefore, we substitute
it with a tractable upper bound.
Recall the fact that the convex function $-g(\cdot)=-\log_2|(\cdot)|$ is  lower
bounded by its first-order Taylor expansion at any local points~\cite{Goldsmith:2017:JSP}. Accordingly, we can derive a concave lower bound for $R_{k}(\boldsymbol {\tieta })$ in objective function~\eqref{eq:problem_PA4} as
\vspace{-0.5em}
\begin{align}\label{eq:R:low_PA}
R_{k}  (\boldsymbol {\tieta })=& g\left(\boldsymbol {\Omega}_{k}  ( {\tieta })\right)-g(\boldsymbol {\Xi}_{k}  (\boldsymbol {\tieta }))\nonumber\\
\geq& g\left(\boldsymbol {\Omega}_{k}  (\boldsymbol {\tieta })\right)-g\left(\boldsymbol {\Xi}_{k}^{ (0)}\right)\nonumber\\
&\hspace{-1em}
-\mathrm{tr}\left(\left(\boldsymbol {\Xi}_{k}^{ (0)}\right)^{-1}\left(\boldsymbol {\Xi}_{k}  ( {\tieta })-\boldsymbol {\Xi}_{k}^{ (0)}\right)\right) \triangleq\widetilde {R}_{k}^{ \mathrm{L}}( {\tieta }),
 \end{align}
where $\boldsymbol {\Xi}_{k}^{ (0)}=\boldsymbol {\Xi}_{k}  ( {\tieta^{(0)}})$.
Moreover, the
convex upper bounds of ${T}$ can be written as
\vspace{-0.5em}
\begin{align}\label{eq:T:low_PA}
    &\widetilde{T}^{\mathrm{U}}(\boldsymbol {\eta}, \tieta) \triangleq \!\!\! \sum_{k'=1}^K \sum_{m\in\Mkp} \sum_{n \in \Mkp} 0.25 \big[(\eta_{k',\cm} \!+\! \eta_{k', \cn})^2  \nonumber\\
&\hspace{1em} - 2(\eta_{k',\cm}^{(0)} \!-\!\eta_{k', \cn}^{(0)})(\eta_{k',\cm} \!-\! \eta_{k', \cn}) + \big(\eta_{k',\cm}^{(0)} \!-\! \eta_{k', \cn}^{(0)}\big)^2 \nonumber\\
&\hspace{1em}-8\etkpmn^{(0)}\etkpmn+4\big(\etkpmn^{(0)}\big)^2 \big],
\end{align}
where we have used the fact that for given  $a\geq0$ and $b\geq0$ 
\vspace{-0.5em}
\begin{align}
ab \leq& 0.25[(a+b)^2\nonumber\\
&~ -2(a^{(0)}-b^{(0)})(a-b) + (a^{(0)}-b^{(0)})^2].
\end{align}
Moreover, the constraints~\eqref{eq:problem_PA4:c2} and~\eqref{eq:problem_PA4:c4} are still non-convex and needs to be approximated by a convex upper bound. Following the fact that the first-order Taylor approximation
is a global upper estimator of a concave function, we have
\begin{align}\label{eq:R:up_PA}
R_{k}  ( {\tieta })\leq& g\left(\boldsymbol {\Omega}_k^{ (0)}\right)
+\mathrm{tr}\left(\left(\boldsymbol {\Omega}_k^{ (0)}\right)^{-1}\left(\boldsymbol {\Omega}_k  ( {\tieta })-\boldsymbol {\Omega}_k^{ (0)}\right)\right)\nonumber\\
&-g(\boldsymbol {\Xi}_k  ( {\tieta }))\triangleq\widetilde {R}_{k}^{ \mathrm{U}}( {\tieta }),
 \end{align}
where $\boldsymbol {\Omega}_k^{ (0)}=\boldsymbol {\Omega}_k  (\boldsymbol {\tieta^{(0)}})$.
In addition, since $\forall a\geq0$ and $b\geq0$ we have\cite{Tung:JWCOM:2020},
\begin{align}
-ab \leq &0.25 [(a\!-\!b)^2\nonumber\\
&-2(a^{(n)}\!+\!b^{(n)})(a\!+\!b)
+ (a^{(n)}+b^{(n)})^2],   
\end{align}
the convex upper bounds of~\eqref{eq:problem_PA4:c4} can be expressed as
\begin{align}
&\etkpmn^2+0.25 [(\eta_{k',\cm}\!-\!\eta_{k', \cn})^2\!-\!2(\eta_{k',\cm}^{(n)}\!+\!\eta_{k', \cn}^{(n)})\times\nonumber\\
&~~(\eta_{k',\cm}\!+\!\eta_{k', \cn})
+ (\eta_{k',\cm}^{(n)}+\eta_{k', \cn}^{(n)})^2]\leq 0.
\end{align}
We are now ready to handle the Problem~\eqref{eq:problem_PA3} by the sequential optimization framework. More specifically, at iteration $(\iota+1)$, for the given points $\boldsymbol {\eta}^{(\iota)}, \tieta^{(\iota)}$, problem \eqref{eq:problem_PA4} can finally be approximated by
\begin{algorithm}[t]\label{Alg:PA}
\caption{Power Allocation Optimization Given User Association}
\begin{algorithmic}[1]
\State Set $\iota=0$ and choose any feasible ${\boldsymbol {\eta}^{(\iota)}}=\boldsymbol {\eta }^{(0)}$ and $\tieta^{(\iota)}=\tieta^{(0)}$.
\Repeat
    \State Perform power allocation optimization~\eqref{eq:problem_PAn}: ${\boldsymbol {\eta}^{(\iota+1)}}$= \emph{Optimize}$({\boldsymbol {\eta}^{(\iota)}})$.
    \State Increase the super-iteration index: $\iota\leftarrow \iota+1$.
\Until{convergence}
\end{algorithmic}
\end{algorithm}
\begin{subequations}\label{eq:problem_PAn}
\begin{align}
&\min _{\boldsymbol {\eta,\tieta}}\quad \widehat{\mathcal{L}}  (\boldsymbol{\eta}, \tieta) \label{eq:problem_PAn:O1}\\
&\text {st.}~\eta _{k,\cm}\geq 0,\qquad  \forall m, k, \label{eq:problem_PAn:c1}  \\
&\hphantom {\text {st.}}  \widetilde {R}_{k}^{\mathrm{U}}(\tieta)  \leq \log_2 (\Morder), \qquad  \forall k,  \label{eq:problem_PAn:c2}\\
&\hphantom {\text {st.}} \sum\nolimits_{k \in \Km}\eta_{k,\cm}\text {tr}\left( \mathrm {\mathbb {E}}\{ \qQ_{mk}\qQ_{mk}^H\} \right)\leq 1, \qquad  \forall m,  \label{eq:problem_PAn:c3}\\
&\hphantom {\text {st.}} \etkpmn^2+0.25 [(\eta_{k',\cm}\!-\!\eta_{k', \cn})^2\!-\!2(\eta_{k',\cm}^{(n)}\!+\!\eta_{k', \cn}^{(n)})\times\nonumber\\
& (\eta_{k',\cm}\!+\!\eta_{k', \cn})+ (\eta_{k',\cm}^{(n)}+\eta_{k', \cn}^{(n)})^2]\leq 0, \quad, \forall k', m,   n, \label{eq:problem_PAn:c4}
\end{align}
\end{subequations}
where
\begin{align}
&\widehat{\mathcal{L}}  (\boldsymbol{\eta}, \tieta)= -\sum_{k=1}^K  \widetilde {R}_{k}^{ \mathrm{L}}( {\tieta })+\lambda\sum_{k'=1}^K \sum_{m\in\Mk} \sum_{n\in\Mk} 0.25\times\nonumber\\
&\big[(\eta_{k',\cm} \!+\! \eta_{k', \cn})^2 - 2\big(\eta_{k',\cm}^{(0)} \!-\!\eta_{k', \cn}^{(0)}\big)\big(\eta_{k',\cm} \!-\! \eta_{k', \cn}\big) +\nonumber\\
&\big(\eta_{k',\cm}^{(0)} \!-\! \eta_{k', \cn}^{(0)}\big)^2 -8\etkpmn^{(0)}\etkpmn+4\big(\etkpmn^{(0)}\big)^2\big].
\end{align}


The objective function in Problem~\eqref{eq:problem_PAn} is convex and the constraints are also affine, and hence  can be solved using standard convex optimization theory. The resulting  sequential optimization
algorithm for solving~\eqref{eq:problem_PAn} is formulated  in {Algorithm 1}\footnote{ Although this work presents results for an i.i.d. fading channel model, the proposed optimization framework in Algorithm 1 is general and can be extended to accommodate more realistic channel models as required.}. Choosing any feasible points $\{\boldsymbol{\eta}, \tieta\}\in\widehat{\mathcal{F}}$  we solve power allocation problem \eqref{eq:problem_PAn} to obtain its optimal solution  $\boldsymbol{\eta}^*, \tieta^*$ which then be used as the initial points in the next iteration. The process being iterated until an accuracy level of $\varepsilon$ is attained. Using~\cite[Proposition 2]{vu18TCOM} it can be readily showed that {Algorithm 1} will converge to a stationary point, i.e., a Fritz John solution, of problem \eqref{eq:problem_PAn} (hence \eqref{eq:problem_PA1}).
%
\subsection{User Association Given Power Allocation}
Thus far in this section, we  dealt with the power allocation design under fronthaul and per-AP transmit power constraint for any given fixed user association.
We now propose a heuristic user association algorithm  for the given power allocation.
 The fronthaul constraint~\eqref{eq:problem2:C3}  poses an upper bound on the number of users that can associate with each AP $m$, $K_m$, by $\Kmax$ as
\begin{align}\label{eq:Kmax}
&\Ntx K_m \left(\alpha_1 \log_2 (M_{\mathrm{order}}) + \alpha_2\right)\leq \mathrm{FH}_{\mathrm{max}}  \nonumber\\
K_m &\leq \left\lfloor\frac{\mathrm{FH}_{\mathrm{max}}}{ \Ntx\left(\alpha_1 \log_2 (M_{\mathrm{order}}) + \alpha_2\right)}\right\rfloor \triangleq \Kmax,
\end{align}
where $\lfloor \cdot \rfloor$ is the floor function.
\begin{algorithm} [t] 
\caption{User Association Given Power Allocation}\label{alg:cap}
\begin{algorithmic}[1]
\State Let $\Km$ be the set of users associated with  AP $m$. Initialize $\Km=\emptyset, m=1,\cdots, M$, and set $\Kmax$ as maximum number of users that can associate with each AP based on~\eqref{eq:Kmax}.
\For {$m=1:M$}
\State Let,  $\beta_{m, m^{(1)}},
\beta_{m, m^{(2)}}, \cdots, \beta_{m, m^{(K)}}$, where $m^{(k)} \in\{1,\cdots,K\}$, be the ordered
channel gain (large-scale fading coefficient)  from AP $m$ to user $k$,  $k=1, \cdots, K$, in descending order $\beta_{m, m^{(1)}}\geq\beta_{m, m^{(2)}}\geq \cdots \geq \beta_{m, m^{(K)}}$.
\State AP $m$ is associated with $\Kmax \leq K$ users  corresponding to the $\Kmax$ largest large-scale fading coefficients. That is, $\Km$ is set as $\Km = \{ m^{(1)}, \dots, m^{(\Kmax)} \}$.
\EndFor
\end{algorithmic}
\end{algorithm}

In the proposed user association  scheme each AP sorts the channel gains in descending
order and independently selects $\Kmax$ users with the strongest channel gains. The proposed user association scheme is described in Algorithm~\ref{alg:cap}.


\begin{table*}[t]
	\centering
\caption{Fronthaul parameters for the simulation} 
\begin{tabular}{||c c c c ||}
 \hline
 Parameter & Value & Parameter & Value \\ [0.5ex]
 \hline\hline
 Number of OFDM subcarriers,  $N_{\mathrm{subcarrier}}$ & 3264  &    {$N_{\mathrm{Gran}}$}       &136  \\
 \hline
$\mathrm{eCPR_{eff}}$ & 0.85      &$\mathrm{delay_{\mathrm{precoder}}} (\mathrm{delay_{data}})$    & $2\times10^{-4}$($5\times10^{-4}$) s  \\
 \hline
 Number of OFDM symbols, $N_{\mathrm{OFDM}}$ & 14 sym &Number of quantization bits, $N_{\mathrm{bits}}$ & 16 \\
 \hline
\end{tabular}
\label{tab:FronthaulParameter}
\end{table*}
\subsection{Computational Complexity}
Here, we compare the complexity of the beamforming design and power allocation for the UC generalized ZF-based CF-mMIMO system. The details of the computational complexity calculation are presented in Appendix~\ref{app:Computational Complexity Analysis}, where complexity is quantified in terms of the number of real floating-point operations (flops).

\emph{1) Computational Complexity of Beamforming Design:}  The total flop count for the beamforming calculation of CB cluster $\Csm$ is
 $\varpi=NU_m+24\Nrx L^2 U_m C_m^2+ 56\Nrx^2 L U_m^2 C_m +
54\Nrx^3 U_m^3$. It is observed that the  computational complexity of the generalized  CF-mMIMO system at each AP $m$ depends on the number of APs in  its corresponding CB cluster, $C_m$, with $1\leq C_m \leq M$, and the cardinality of its associated user set $\Usm$, $U_m$, with $1\leq U_m \leq K$. The corresponding computational complexity of fully centralized  beamforming design (fully distributed beamforming design) can be obtained by replacing $C_m$ and $U_m$ in $\varpi$ with $M$ and $K$ ($1$ and $\Kmax$), respectively.

\emph{2) Computational Complexity of Power Allocation:} The computational complexity of  power allocation for the generalized CF-mMIMO system with $C$ CB is given by
$\mathcal{O}(\iota\sqrt{M +K + KC^2}(KC(2C+1) + M +K )(KC(C+1))^2)$. The CF-mMIMO system with fully distributed beamforming design exhibits the highest power allocation computation complexity. This is because its corresponding   problem involves computing up to $MK$ power allocation coefficients. 

We emphasize that the power allocation optimization in Algorithm~1 is scalable, with scalability ensured through careful consideration of computational complexity, adaptability to various system setups, and fronthaul constraints. Similarly, the proposed user association in Algorithm~2 is simple, low-complexity, and efficient. Both algorithms rely on statistical CSI, requiring only large-scale channel statistics, which change slowly over time and remain constant across frequencies. Unlike small-scale fading coefficients, which vary rapidly, large-scale fading evolves more gradually. Therefore, when the system remains static (with unchanged users) for extended periods, power allocation and user association only need to be performed once and can be applied to all subcarriers across multiple frame durations. This significantly reduces the need for frequent recalculations, while bemforming, based on instantaneous CSI, remains the dominant factor in computational complexity.   Among the considered beamforming designs, the distributed design has the lowest complexity in calculating beamforming matrices and the least fronthaul overhead for CSI exchange, while the centralized design has the highest computational complexity and fronthaul overhead.


\section {Numerical Results}\label{sec:Numerical Results}
We consider a CF-mMIMO system where the APs and
users are randomly distributed in a square of $1 \times 1$ km${}^2$,
whose edges are wrapped around to avoid the boundary effects. 
Each AP can serve up to $K_{\mathrm{max}}$ users out of a set of $K$ users. The value of parameter $K_{\mathrm{max}}$  depends on the system and fronthaul parameters and is determined based on~\eqref{eq:Kmax}.  


\subsection{Simulation Setup and Parameters}
 The large-scale fading coefficient $\beta_{mk}$ includes the path loss
and shadow fading, according to
\begin{equation} \label{eq:Beta}
\beta_{mk} =10^{\frac{\text{PL}_{mk}}{10}}  10^{\frac {\sigma _{sh} ~y_{mk}}{10}}, 
\end{equation}
where  the second term models the shadow fading  with standard deviation $\sigma_{sh} = 4$ dB and $y_{mk} \sim \mathcal{CN}(0, 1)$, while   $\text{PL}_{mk}$ (in dB) is calculated as 
\begin{align}\label{eq:PL}
&\hspace {-0.3pc}\text {PL}_{mk}=\! 
\begin{cases} \!-L-35\log_{10}(d_{mk}),\qquad \qquad \quad & d_{mk} > d_{1},\!\! \\ \!-\!L-\!15\!\log _{10}(d_{1})\!-\!\!20 \log_{10}(d_{mk}),~\!&\! d_{0}\! < \!d_{mk}\!\!\le \!\! d_{1},\!\!\\ \!-L-15 \log_{10}(d_{1})-20 \log _{10}(d_{0}),&d_{mk} \le d_{0}, 
\end{cases} 
\end{align}
with $L = 46.3 + 33.9 \log_{10}(f)-13.82 \log_{10}(h_{\mathrm{AP}} )-
(1.1 \log_{10}(f)- 0.7)h_{\mathrm{U}} + (1.56 \log_{10}(f)- 0.8)$. In addition, $f$
is the carrier frequency (in MHz),    $d_{mk}$  denotes the distance between AP $m$ and user $k$ (in m), $h_{\mathrm{AP}}$ and $h_{\mathrm{U}}$  are the antenna heights at the AP and at the user (in m), respectively.  Here, we choose $d_0 = 10$ m, $d_1 = 50$ m,  $h_{\text{AP}} = 15$ m and $h_{\text{U}} = 1.65$ m. These parameters resemble those in \cite{hien:2017:wcom}.   
The maximum transmit power for training  pilot sequences and each AP is $100$ mW, while the  the noise power is $\sigma^2_n=-92$ dBm.
In what follows,  unless otherwise stated, we choose the  fronthaul parameters summarized in  Table~\ref{tab:FronthaulParameter}, and set $M_{\mathrm{order}}=2^9$.  In addition,  we consider $100$ MHz bandwidth with $30$ kHz subcarrier spacing which corresponds to $N_{\mathrm{subcarrier}}=3264$  as in Table~\ref{tab:FronthaulParameter}.
Moreover, we select $\tau_\mathrm{c} = 2000$ samples, corresponding to a coherence bandwidth of
$200$ kHz and a coherence time of $10$ ms.

We compare the performance of the following cases:
\begin{itemize}
\item $C=1$:   In this case there is $1$ CB cluster consisting of all the APs in the network, which can be considered as a CF-mMIMO system with fully centralized CB. 

\item $C=2$: In this case, APs are divided into $C=2$ disjoint CB clusters.

\item $C=4$: In this case,  APs are divided into $C=4$ disjoint CB clusters.

\item $C=M$:  In this case, we have  $C=M$ CB   clusters with cardinality 1. This case corresponds to the CF-mMIMO system with fully distributed beamforming. 
\end{itemize}
For CB clustering, we consider a simple distributed CB architecture where APs are divided into $C$ disjoint CB clusters based on their geographical positions. Moreover, the notations  OPA  and EPA describe the performance achieved by the proposed  power allocation {Algorithm~1} and  low-complexity heuristic  power control~\cite{nayebi:2017:wcom}, as outlined as follows, respectively. With~\eqref{eq:DL-x1} and~\eqref{eq:P_const}, the per-AP power constraint for AP $m$ in CB cluster $\cm$  can be written as
$\sum_{k \in \Km}\eta_{k, \cm}\text {tr}\left( \mathrm {\mathbb {E}}\{ \qQ_{mk}\qQ_{mk}^H\} \right)\leq 1$. Accordingly, in EPA,  the power coefficient used by AP $m$ for transmission to user $k$ is calculated as
 \begin{equation}~\label{EPA:Benchmark}
\eta_{k, \cm}=\frac{1}{\max_m(\sum_{k \in \Km}\delta_{mk})},~~~~  \forall k \in \Km , m \in \cm,
\end{equation}
where $\delta_{mk}\triangleq\text {tr}\left(\mathrm {\mathbb {E}}\{ \qQ_{mk}\qQ_{mk}^H\} \right)$.


\begin{figure}[t]
\centering
  \begin{subfigure}[t]{0.48\textwidth}
    \includegraphics[width=\textwidth]{  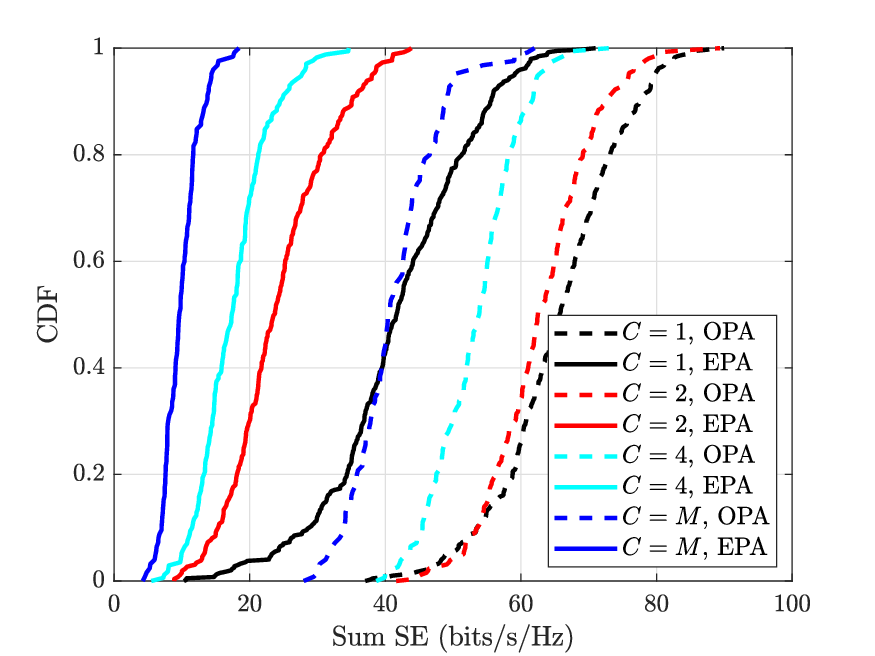}
    \caption{CDF of the sum SE}
    \label{Fig3a}
  \end{subfigure}
  \begin{subfigure}[t]{0.48\textwidth}
    \includegraphics[width=\textwidth]{  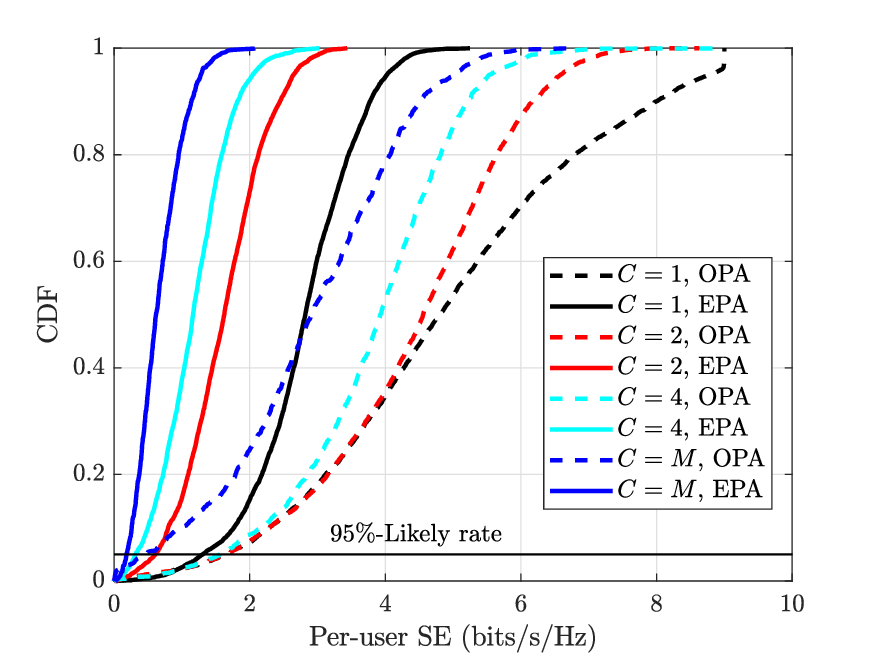}
    \caption{CDF of the per-user SE}
    \label{Fig3b}
  \end{subfigure}
  \caption{Comparison between the sum SE and per-user SE achieved by the proposed power allocation and user association and baseline schemes for different number of CB clusters.  Here, $M=10, N=2,\bar{N}=1, \mathrm{FH}_{\mathrm{max}}=14$ Gbps, and $L=22$. }
  	\vspace{0.6em}
\label{fig:Fig3}
\end{figure}

\begin{figure}[t]
	\centering
	\includegraphics[width=0.48\textwidth]{  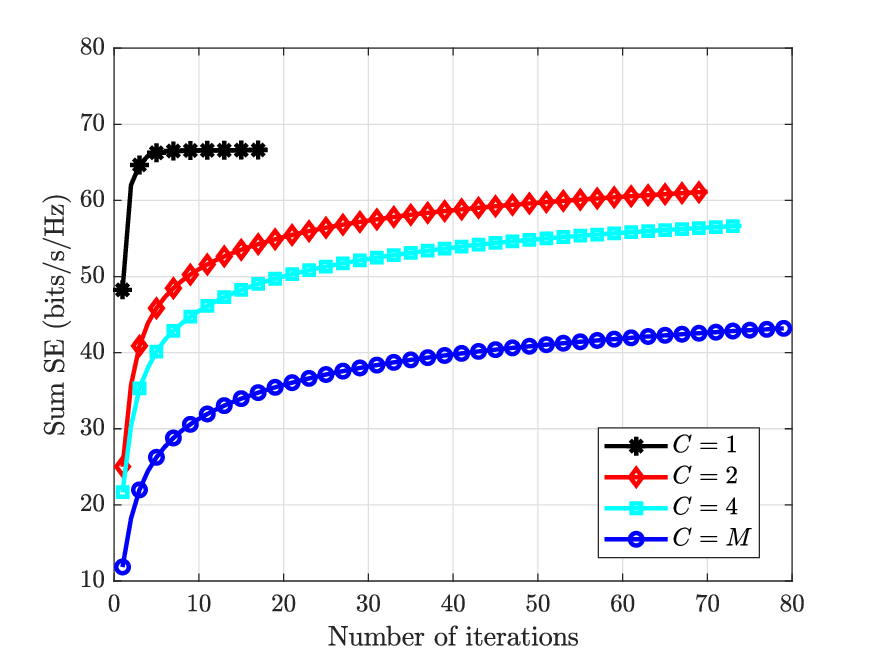}
	\vspace{-0.6em}
	\caption{Convergence of the proposed Algorithm 1, where $M=10$, $K=15$,   $N=2$, $\bar{N}=1$, $L=22$, and $\mathrm{FH}_{\mathrm{max}}=14$ Gbps.}
	\label{fig:Fig10}
    	\vspace{0.6em}
\end{figure}
\subsection{Results and Discussions}
 \emph{1)  Power Allocation and User Association:} In Fig.~\ref{fig:Fig3} we evaluate the performance of
the proposed power allocation and user association in the fronthaul-aware downlink  CF-mMIMO system. 
 Numerical results lead to the following conclusions.
\begin{itemize}
\item The proposed  power allocation and user association solution resulting from the optimization Algorithm 1, yields significant
sum-SE and per-user SE performance increase against EPA for different cases.  More specifically, it
provides sum-SE performance gains of up to $59\%, 173\%, 252\%$, and $312\%$, 
over EPA for $C=1$, $C=2$, $C=4$, and $C=M$ cases, respectively, which highlights the advantage of our power optimization solution
over heuristic ones. 
\item   It is observed that the performance gain is minimal when $C=1$, whereas it is maximal when $C=M$. The  reason is twofold: 1) primarily due to the fact that when $C=1$, the fully centralized CB handles all interference cancellation, leaving little room for further improvement in sum-SE performance. However, for larger values of $C$, the residual interference can be effectively mitigated through a carefully designed power allocation scheme, leading to significant performance enhancements when power allocation is optimized; 2) secondarily because when power
allocation is optimized,  increasing the number of CB clusters  provides a
greater flexibility. In fact, in the system with $C=1$ CB cluster, we optimize up to $K$ power allocation coefficients. However, in the system with $C$ CB clusters, we can optimize the  sum SE over $CK$ transmit powers. Therefore, the solution space is expanding by increasing the number of   CB clusters.


\item The sum-SE  performance gap between $C=M$  and $C=1$,  decreases with OPA. For example,  the sum-SE performance loss of $C=M$ case with respect to $C=1$ case is $75\%$ and $38 \%$ under EPA and OPA strategies, respectively.
\item  From the $95\%$-likely SE performance perspective, the CF-mMIMO system relying on OPA provides nearly identical performance for the cases of $C=1$, $C=2$, and $C=4$, which is three times higher than that of the $C=M$ case. This result highlights the efficiency of OPA in maintaining consistent SE performance across varying numbers of CB clusters. 
\end{itemize}
 
Figure~\ref{fig:Fig10} illustrates the convergence behaviour of the proposed approach for different numbers of CB clusters. The algorithm is computationally efficient, with the number of iterations being minimal when $C = 1$ and maximal when $C = M$. This is because, for $C = 1$, there are $K$ power allocation variables, while for $C$ CB clusters, there are $CK$ transmit power variables.

\begin{figure}[t]
	\centering
	\includegraphics[width=0.48\textwidth]{  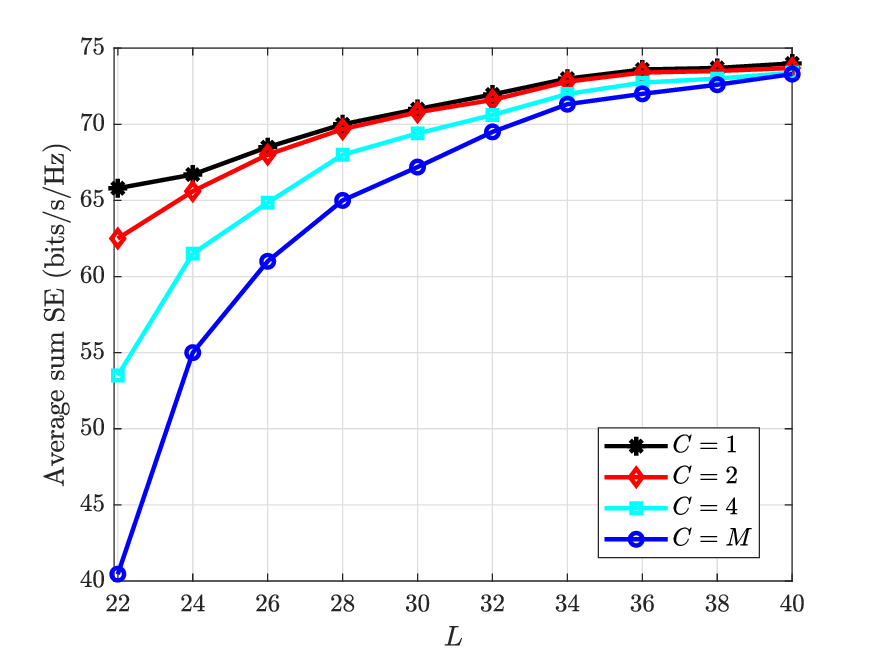}
	\vspace{-1em}
	\caption{Average sum SE  versus number of antennas per AP, $L$, where $M=10$, $K=15$, $\mathrm{FH}_{\mathrm{max}}=14$ Gbps,  $N=2$, and $\bar{N}=1$.}
	\label{fig:Fig4}
    	\vspace{-1.6em}
\end{figure}
\begin{figure}[t]
	\centering
	\includegraphics[width=0.48\textwidth]{  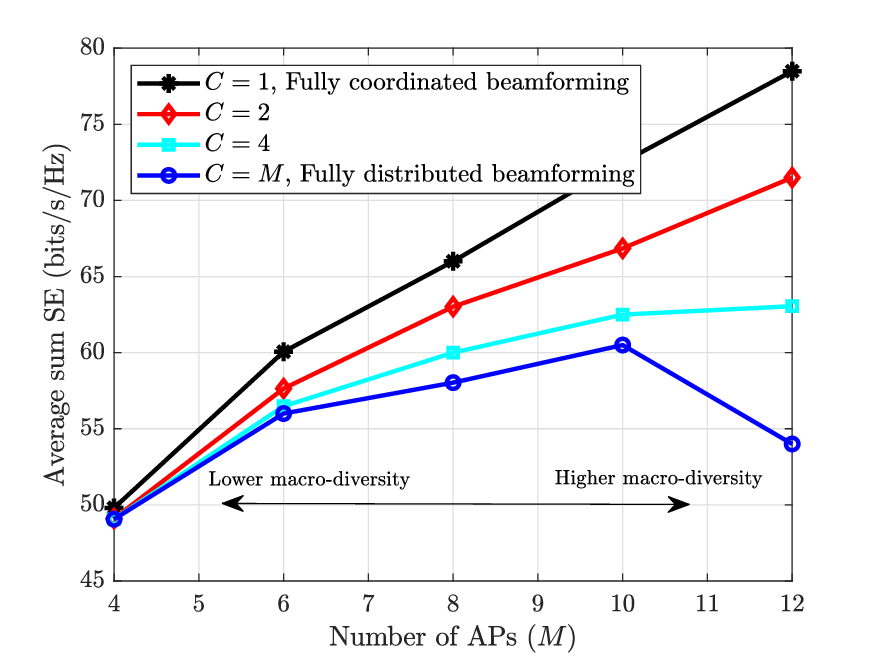}
	\vspace{-0.8em}
	\caption{Average sum SE  versus number of APs, $M$, where $ML=240$, $K=15$, $\mathrm{FH}_{\mathrm{max}}=14$ Gbps, and $N=\bar{N}=1$.}
	\label{fig:Fig5}
    	\vspace{0.6em}
\end{figure}


\emph{2) Effect of the Number of Antennas Per AP:} Figure~\ref{fig:Fig4} shows the average sum-SE performance achieved by CF-mMIMO system for different number of transmit antennas at the AP. Increasing $L$ has two effects on the
sum-SE performance, namely, (i) increasing the
diversity and array gain, and (ii) decreasing $K_{\mathrm{max}}$ based on fronthaul constraint~\eqref{eq:Kmax}. The
former effect becomes dominant under distributed scheme with $C=M$, which leads to a significant improvement in the SE performance. However, the latter effect becomes dominant under centralized CB scheme with $C=1$. As a result, the performance gap between  centralized and distributed scheme
decreases with increasing $L$. Later, we will investigate the performance of the CF-mMIMO system under different fronthaul constraint values, as shown in Fig.~\ref{fig:Fig6}.  Simulation results in Fig.~\ref{fig:Fig4} also
indicate that: 1) for high number of transmit antennas, CF-mMIMO with distributed beamforming relying on the proposed optimization solution has only negligible performance penalty, despite avoiding the high
computational complexity for centralized CB and heavy  overhead used for acquisition of CSI among large CB coordination sets. This result highlights the advantage of our  optimization framework in reducing the need for beamforming coordination among a large number of APs;  2) a balance between performance, complexity, and scalability can be achieved by selecting an appropriate number of CB clusters, $C$, enabling partial cooperation to enhance performance while maintaining manageable complexity and fronthaul demands.

\emph{3) Effect of the Number of APs:} 
Figures~\ref{fig:Fig5} depicts the average  sum SE as a function of the number of APs for systems having the same total numbers of service antennas, i.e., $L M = 240$, but different number of APs. From this figure, we have
the following observations
\begin{itemize}
    \item Under fully coordinated beamforming and/or relatively large cooperation sizes, distributing antennas trivially achieves better performance thanks to added macro diversity on top of high-levels of interference cancellation due to CB design.

    \item On the other hand, for $C=M$ (fully distributed beamforming), a non-trivial trade-off between (per AP) beamforming capability and macro diversity should be achieved; i.e., further distributing antennas without beamforming coordination is detrimental.
    
\end{itemize}

\begin{figure}[t]
	\centering
	\includegraphics[width=0.48\textwidth]{  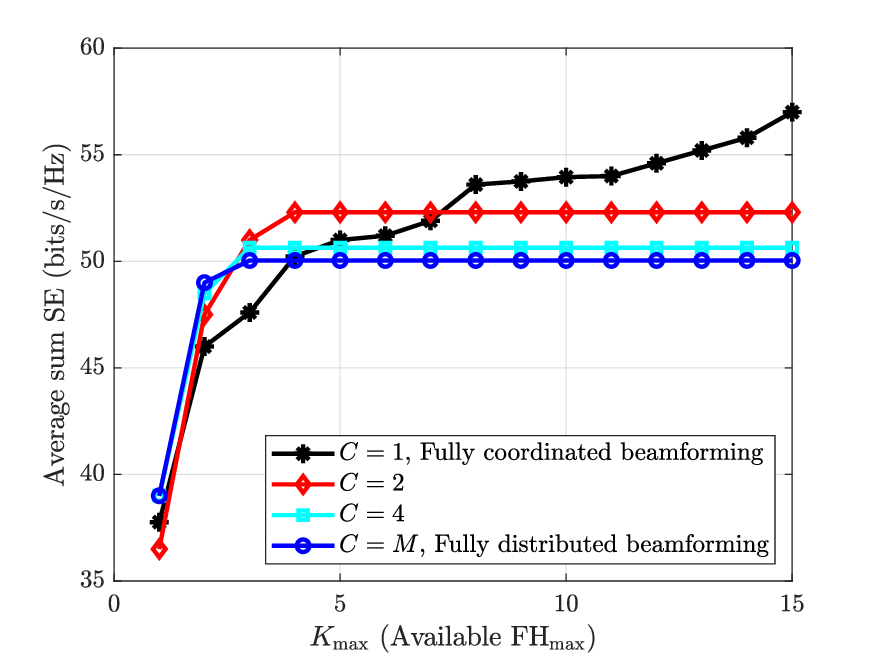}
	\vspace{-0.6em}
	\caption{Impact of the $\mathrm{FH}_{\mathrm{max}}$ on the sum SE, where $M=10$, $K=15$,  $N=2$, $\bar{N}=1$, and $L=12$.}
	\label{fig:Fig6}
    	\vspace{.8em}
\end{figure}

\emph{4) Effect of the Maximum Fronthaul Capacity:} In Fig.~\ref{fig:Fig6} we  investigate the effect of the maximum available fronthaul capacity, $\mathrm{FH}_{\mathrm{max}}$,  on the performance of the  CF-mMIMO system with the  proposed  optimization scheme under the fronthaul constraint.  Different values of $\mathrm{FH}_{\mathrm{max}}$, correspond to different values of $\Kmax$, which are calculated based on~\eqref{eq:Kmax}.  Upon increasing $\mathrm{FH}_{\mathrm{max}}$ the upper bound on the maximum number of users that each AP can serve, i.e., $\Kmax$, increases. We note that in this figure, for each value of $\Kmax$, we change the  number of   assigned users to the APs, $K_m$, from $1$ to $\Kmax$, and then choose  the value that results in the maximum sum SE as
\vspace{-0.7em}
\begin{align}
    K_m^*=\argmax_{K_m \in [1,\cdots,\Kmax]} \sum_{k \in \mathcal{K}}R_k(K_m).
\end{align}
From this figure, we have
the following observations.
\begin{itemize}
        \item The performance of CF-mMIMO with centralized CB notably increases with increasing $\Kmax$. The main reason for this is the CB designs' capability to mitigate multi-user interference among users within the same CB cluster, which allows APs to select higher number of users to serve.
        \item For the considered CF-mMIMO system with distributed beamforming design $\Kmax=3$ is enough and increasing $\Kmax$ from $3$ to $15$ will not increase the  sum-SE performance.  This is due to the fact that distributed beamforming design exploits all degrees of freedom to mitigate its own interference to its assigned users, but not interference from other APs which is high for larger values of  $\Kmax$.
    \item  In regimes with higher values of $\Kmax$, which result from higher available $\mathrm{FH}_{\mathrm{max}}$, centralized CB design outperforms other cases, while increasing the number of CB clusters leads to a reduction in the sum-SE performance.     We also note that for the system with highly limited $\mathrm{FH}_{\mathrm{max}}$, distributed beamforming design  having the ability to provide a better performance/implementation complexity trade-off  is  undoubtedly a better choice.
    
        \item The relative performance gap between different cases changes when $\Kmax<8$. This change occurs because the CF-mMIMO system experiences varying levels of trade-offs between degrees of freedom for power optimization and inter-user interference cancellation for different numbers of CB clusters.

         \item Finally, our results indicate that under practical fronthaul limitations $\mathrm{FH}_{\mathrm{max}}$, the proposed power allocation combined with simple distributed CB architecture can even outperform a computationally-heavy fully centralized CB when $\Kmax \leq 4$.
\end{itemize}
\begin{figure}[t]
	\centering
	\includegraphics[width=0.48\textwidth]{  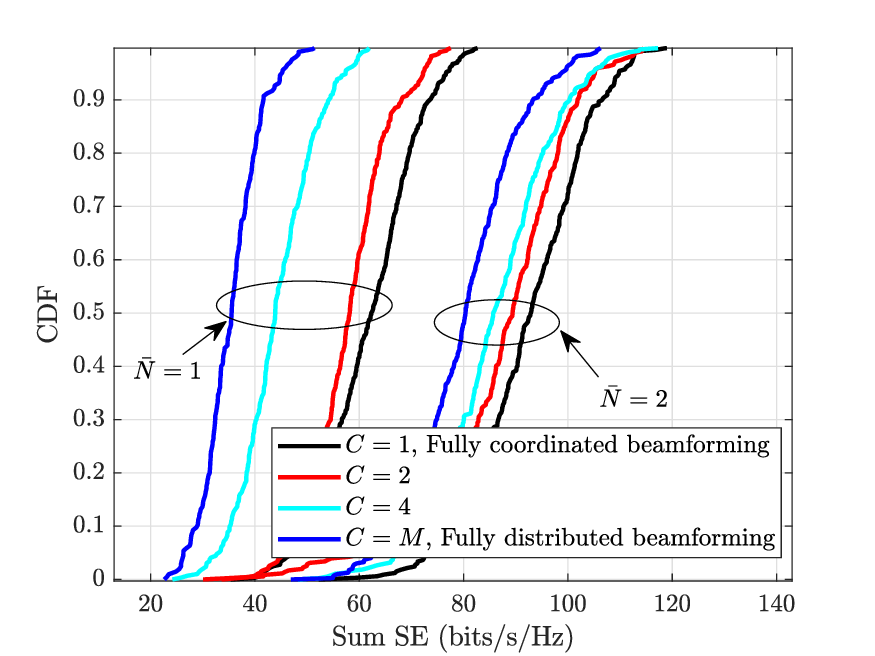}
	\vspace{-0.6em}
	\caption{CDF of the sum SE, where $M=10$, $K=15$,  $N=2$, $L=18$, and $\mathrm{FH}_{\mathrm{max}}=18$ Gbps.}
	\label{fig:Fig7}
        	\vspace{.8em}
\end{figure}


\begin{figure}[t]
	\centering
	\includegraphics[width=0.48\textwidth]{  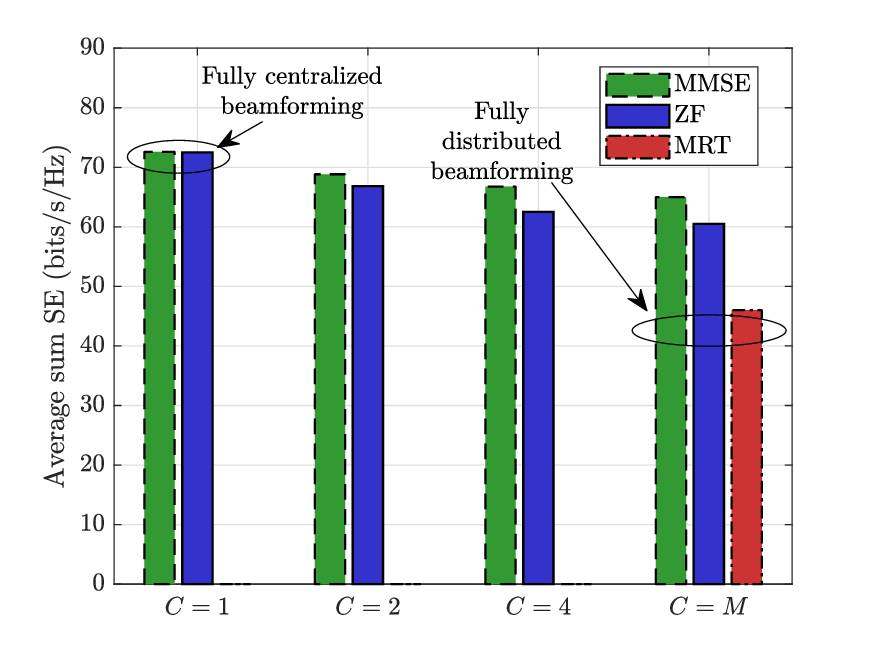}
	\vspace{-0.6em}
	\caption{Comparison between the sum SE achieved by different beamforming schemes, where $M=10$, $K=15$,  $N=\bar{N}=1$, $L=24$, and $\mathrm{FH}_{\mathrm{max}}=14$ Gbps.}
	\label{fig:Fig8}
    	\vspace{-0.6em}
\end{figure}
\begin{figure}[t]
	\centering
	\includegraphics[width=0.48\textwidth]{  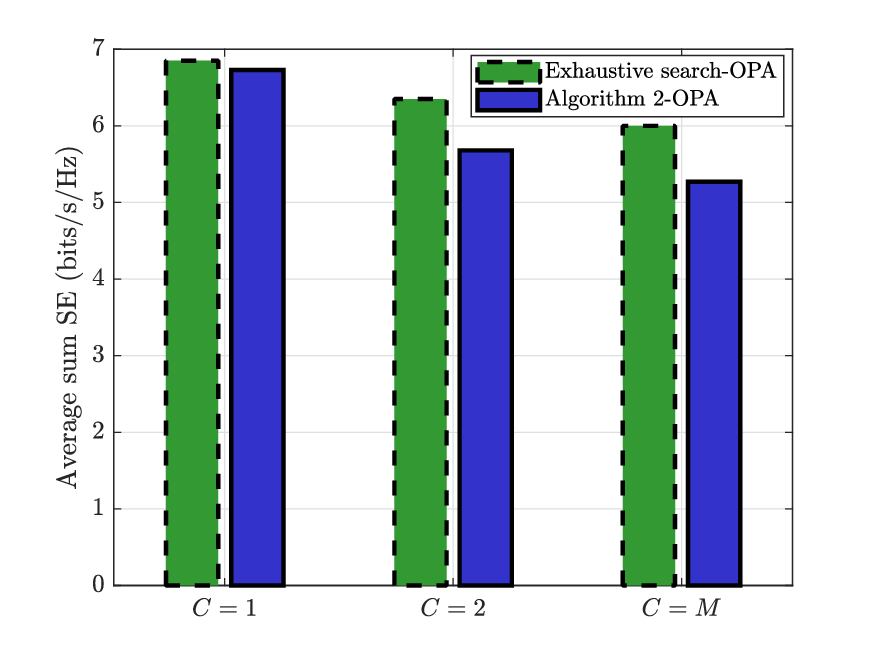}
	\vspace{-0.6em}
	\caption{Comparison between the heuristic user association Algorithm 2 and exhaustive search, where $M=4$, $K=3$,  $N=2$, $\bar{N}=1$, $L=8$, and $\mathrm{FH}_{\mathrm{max}}=3$ Gbps.}
	\label{fig:Fig9}
        	\vspace{.6em}
\end{figure}

\emph{5) Effect of the Number of  Downlink Data Streams:}
In Fig.~\ref{fig:Fig7}, we investigate the effect of number of downlink data stream $\Ntx$ on the sum-SE performance of the CF-mMIMO system for different number of CB clusters. We can see that using multiple data streams at the users greatly improves the sum SE  for all cases. This is an interesting observation, as increasing $\Ntx$ has two impacts on the CF-mMIMO system: 1) enhancing the multiplexing gain and 2) increasing the amount of data traffic to be fronthauled between BBHs and BBLs, consequently reducing $\Kmax$ based on the fronthaul constraint~\eqref{eq:Kmax}. The increase in multiplexing gain dominates the sum-SE performance.

\emph{5) Effect of the Beamforming Scheme:}
In Fig.~\ref{fig:Fig8}, we compare the performance of a CF-mMIMO system using ZF-based beamforming design and the MMSE scheme for different numbers of CB clusters. Moreover, we present results for the fully distributed MRT scheme. Simulation results indicate that the performance gap between ZF and MMSE is minimal, particularly for a smaller number of CB clusters. In the fully distributed case, MMSE and ZF beamforming schemes achieve sum-SE performance gains of $39\%$ and $30\%$ over the MRT scheme, respectively.  

Finally, in Fig.~\ref{fig:Fig9}, we benchmark the proposed heuristic user association algorithm (Algorithm 2) against an optimal approach using exhaustive search. A small system setup comprising 4 APs and 3 users is considered to enable the implementation of exhaustive search. The results demonstrate that the proposed heuristic approach achieves near-optimal performance in the fully centralized case with $C=1$, while significantly reducing computational complexity. As the number of CB clusters increases, the performance loss of the heuristic user association algorithm grows; for the fully distributed case, which corresponds to the highest number of CB clusters, the performance loss remains below  $14\%$.
\section{Conclusions}
In this paper, we have provided a comprehensive analytical framework for the fronthaul limited CF-mMIMO systems with multiple-antenna users and multiple-antenna APs relying on  both CB and UC clustering.
The proposed CF-mMIMO framework is very general and can cover different CF-mMIMO scenarios with different performance-overhead and complexity trade-offs. We proposed a large-scale fading-based  optimization approach  to maximize the sum SE   under a total transmit power constraint  and fronthaul constraint. Our optimized approach demonstrated significant sum-SE gains compared to the baseline, particularly for fully distributed CB designs. Specifically, it achieved sum-SE performance gains of up to $59\%, 173\%, 252\%$, and $312\%$,  over EPA for the cases of  $C=1$, $C=2$, $C=4$, and $C=M$, respectively.
Insights from the simulation results demonstrated that  the number of APs, the number of antennas, and the maximum fronthaul capacity are the key factors determining the extent to which decentralizing the beamforming design,  such as by reducing the number of CB clusters,  enhances the system performance. For example, if either the number of transmit/receive antennas is large or the available fronthaul is small, distributed beamforming design offers  a better performance/implementation complexity trade-off. Nevertheless, the sum-SE  performance improvement of the centralized CB design is more pronounced for  higher values of available $\mathrm{FH}_{\mathrm{max}}$. Future work will explore user mobility, Doppler effects, and hardware impairments in fronthaul-limited user-centric CF-mMIMO systems with CB. Moreover, more realistic channel models will be investigated to enhance practical applicability and refine performance predictions.

\appendices

\section{Computational Complexity Analysis}  \label{app:Computational Complexity Analysis}
The complexity is counted as the number of real flops where a real addition, multiplication, or division is counted as one flop. Also, a complex addition
and multiplication have two flops and six flops, respectively.

\emph{1) Computational Complexity of Different Beamforming Designs:}
The computational complexity of the proposed beamforming schemes is elaborated as follows.
Beamforming matrix calculation of a given AP $m$ in \eqref{eq:Q2} for the CB cluster $\Csm$ can be divided into two parts, i.e., the calculation of
SVD of channel matrix $\widehat {\qG}_{m}$ as $\widehat {\qG}_{m} = \qU_m \boldsymbol{\Sigma}_m \qV^H_m$ and that of $\qU_m \boldsymbol{\Sigma}^{-1}_m \qV^H_m$.

 \begin{itemize}
\item The flop count for SVD of complex-valued $(L C_m)\times (\Nrx  U_m)$ channel matrix $\widehat {\qG}_{m}$ in~\eqref{eq:G}  is
$24\Nrx L^2 U_m C_m^2+ 48\Nrx^2 L U_m^2 C_m +
54\Nrx^3 U_m^3$ by treating every operation as complex multiplication~\cite{Shen:2012:JSP}.
\item For the second part, calculating $\boldsymbol{\Sigma}^{-1}_m$ needs $NU_m$ real divisions and calculating $\qU_m \boldsymbol{\Sigma}^{-1}_m$ with given $\boldsymbol{\Sigma}^{-1}_m$ needs $2LNC_mU_m$  real multiplications. Then multiplying $\qU_m \boldsymbol{\Sigma}^{-1}_m$ with  $\qV^H_m$ needs  $8LN^2C_mU_m^2-2LNC_mU_m$ flops.
    \end{itemize}
Hence, the total flop count for the beamforming calculation  is
 $\varpi=NU_m+24\Nrx L^2 U_m C_m^2+ 56\Nrx^2 L U_m^2 C_m +
54\Nrx^3 U_m^3$.

\emph{2) Complexity of  Power Allocation:}
In each iteration of Algorithm 2, the computational complexity of solving problem~\eqref{eq:problem_PA2} is
$\mathcal{O}( \sqrt{n_{l1}+\\n_{q1}} (n_{v1} + n_{l1} + n_{q1})n_{v1}^2)$,  where $n_{v1}=KC(C+1)$ is the number of real-valued scalar decision variables with $C$ is the number of CB clusters, $n_{l1}=M +K$ denotes the number of linear constraints, and  $n_{q1}=KC^2$ is the number of quadratic constraints~\cite{Vincent:TWC:2017}. Therefore, the number of flops required by the
Algorithm 1 is $\mathcal{O}( \sqrt{M +K+ KC^2}(KC(2C+1) + M +K)(KC(C+1))^2)$.
\bibliographystyle{IEEEtran}
\bibliography{IEEEabrv,Ref_CellFreeFronthaul}

\balance

\end{document}